\documentclass[a4paper,UKenglish,cleveref, autoref, thm-restate]{lipics-v2021}

\title{Infinitary Cut-Elimination for Non-Wellfounded Parsimonious Linear Logic} 

\author{Matteo Acclavio}{
	University of Sussex, UK \and \url{matteoacclavio.com} 
}{
}{
	https://orcid.org/0000-0002-0425-2825
}{%
	The author was partially supported by Villum Fonden, grant no. 50079
}

\author{
	Gianluca Curzi
}{
	University of Gothenburg, Sweden \and \url{http://gianlucacurzi.com/}
}{
}{
	https://orcid.org/0000-0001-8746-1704
}{%
	This work was supported by the Wallenberg Academy Fellowship Prolongation project “Taming Jörmungandr: The Logical Foundations of Circularity” (project reference 251080003), and by the VR starting grant “Proofs with Cycles in Computation” (project reference 251088801)
}

\author{
	Giulio Guerrieri
}{
	University of Sussex, UK
	\and \url{https://pageperso.lis-lab.fr/~giulio.guerrieri/}%
}{
}{
	https://orcid.org/0000-0002-0469-4279
}{%
}

\authorrunning{M. Acclavio, G. Curzi and G. Guerrieri} 

\Copyright{Acclavio, Curzi, Guerrieri} 



\category{} 

\relatedversion{} 

\nolinenumbers 

\EventEditors{John Q. Open and Joan R. Access}
\EventNoEds{2}
\EventLongTitle{42nd Conference on Very Important Topics (CVIT 2016)}
\EventShortTitle{CVIT 2016}
\EventAcronym{CVIT}
\EventYear{2016}
\EventDate{December 24--27, 2016}
\EventLocation{Little Whinging, United Kingdom}
\EventLogo{}
\SeriesVolume{42}
\ArticleNo{23}

\usepackage[utf8]{inputenc}
\usepackage{amssymb}
\usepackage{amsmath}
\usepackage{graphics}
\usepackage{graphicx}
\usepackage{amsthm}

\usepackage{mathtools}
\usepackage{hyperref}
\usepackage{virginialake}
\vlnostructuressyntax
\usepackage{tikz-cd} 
\usepackage{cmll} 
\usepackage{cleveref}
\usepackage{url}
\usepackage[textsize=footnotesize]{todonotes}
\usepackage{thmtools,thm-restate}
\usepackage{subfloat}
\usepackage{enumerate}

\usepackage{pifont}
\usepackage{adjustbox}
\usepackage{xspace}
\usepackage[notxmath]{proofzilla}

\def\queq{\quad=\quad}

\newcommand{\weightcp}[1]{\mathsf{C}(#1)}
\newcommand{\heightcut}[1]{\mathsf{H}(#1)}

\newcommand{\cf}[1]{\mathsf{cf}{(#1)}}

\newcommand{\prun}[2]{\lfloor {#1} \rfloor_{#2}}
\newcommand{\base}[1]{\mathsf{base}(#1)}

\def\Sys{\mathcal S}

\def\wnderr{\wn\mathsf{der}}

\newcommand{\dfn}{\coloneqq}

\renewcommand{\emptyset}{\varnothing}

\def\quand{\quad \mbox{and} \quad}

\def\defn#1{\textbf{#1}}

\def\set#1{\{#1\}}
\def\Set#1{\left\{#1\right\}}
\def\intset#1#2{\set{#1,\ldots, #2}}

\def\Setdef#1#2{\Set{\!\!\begin{array}{c|c}#1 & #2\end{array}\!\!}}
\def\seq#1{(#1)}

\definecolor{mygreen}{rgb}{0, 0.6, 0}

\definecolor{airforceblue}{rgb}{0.36, 0.54, 0.66}

\usepackage{marginnote}

\def\cneg#1{{#1}^\lbot}

\def\atoms{\mathcal A}

\newcommand{\Nat}{\mathbf{N}}

\def\axr{\mathsf{ax}}
\def\cutr{\mathsf{cut}}

\def\wrule{\mathsf{w}}
\def\crule{\mathsf{c}}
\def\brule{\mathsf{b}}

\def\wnbrule{\wn\brule}
\def\ocbrule{\oc\brule}
\def\wnwrule{\wn \wrule}
\def\wncrule{\wn \crule}
\def\ocwrule{\oc \wrule}

\def\prule{\mathsf{\oc p}}
\def\derr{\mathsf{\wn d}}
\def\wnderr{\derr}

\def\digr{\mathsf{\wn\wn d}}

\def\fprule{\mathsf{f\prule}}
\def\cprule{\mathsf{c\prule}}
\def\rrule{\mathsf{r}}

\def\trule{\mathsf{t}}

\def\unit{\mathbf{1}}

\newcommand{\zero}{\mathsf{hyp}}
\newcommand{\zerorule}[1]{
	\vlinf{\zero}{}{{#1}}{}
}

\def\rclr#1{{\color{red}#1}}

\def\mathacronym#1{\textsf{#1}}

\newcommand{\nuprule}{\mathsf{ib}\prule} 
\newcommand{\nwbox}{\mathacronym{nwb}\xspace} 
\newcommand{\nwboxes}{\mathacronym{nwb}s\xspace}
\newcommand{\nwpromotion}{\mathfrak{S}}
\newcommand{\nwpstream}[3]{\seq{ {#1},{#2}, \ldots,\allowbreak {#3}, \ldots } }

\newcommand{\derstream}[1]{\cprule_{\seq{ #1, \ldots}}}

\newcommand{\der}{\mathcal{D}}
\newcommand{\dD}{\der}

\newcommand{\zeroder}{\der_{\!\lightning}}
\newcommand{\wnder}{\der_{\wn}}
\newcommand{\trueder}{\cod\true}
\newcommand{\falseder}{\cod\false}
\newcommand{\digder}{\der_{\digr}}

\newcommand{\vlhyE}[1]{\global\setbox\vlhybox=\hbox{$#1$}%
	\vlhyaux{\box\vlhybox}}%
\newcommand{\vldr}[3][8]{\vltr{#2}{#3}{\vlhyE{}}{\vlhyE{\hspace{#1pt} }}{\vlhyE{}}}

\newcommand{\pll}{\mathsf{PLL}}
\newcommand{\nwpll}{\pll^{\!\infty}}
\newcommand{\nupll}{\mathsf{wr}\nwpll}
\newcommand{\cpll}{\mathsf{r}\nwpll}
\newcommand{\opll}{\mathsf{o}\nwpll}
\newcommand{\ppll}{\mathsf{p}\nwpll}
\newcommand{\wppll}{\mathsf{w}\ppll}
\newcommand{\dpll}{\mathsf{nu}\pll}

\newcommand{\mell}{\mathsf{MELL}}
\newcommand{\nwmell}{\mell^{\!\infty}}

\def\sysX{\mathsf X}

\def\proves#1{\vdash_{#1}}

\newcommand{\setDer}{\mathsf{X}}

\def\parsimony{finite expandability\xspace}
\def\parsimonious{finitely expandable\xspace}

\def\FE{\parsimonious}
\def\octhread{$\oc$-thread\xspace}

\def\wnthread{$\wn$-thread\xspace}

\def\prog{progressing\xspace}
\def\wprog{weak-progressing\xspace}

\def\progness{progressing criterion\xspace}
\def\Prog{Progressing\xspace}

\newcommand{\Nodes}{\mathcal{V}}

\newcommand{\Nset}{\mathbb{N}}

\def\cproj#1{\cprojp{\left(\!#1\!\right)\!\!}}
\def\fproj#1{\fprojp{\left(\!#1\!\right)\!\!}}
\def\finproj#1{\finprojp{\left(\!#1\right)}}
\newcommand{\cprojp}[1]{{#1}^{\bullet}}
\newcommand{\fprojp}[1]{{#1}^{\circ}}
\newcommand{\finprojp}[1]{{#1}^{f}}
\def\ptom#1{\left(#1\right)^\spadesuit}
\def\ptomp#1{#1^\spadesuit}

\newcommand{\size}[1]{|#1|}

\newcommand{\nestl}[2]{\mathbf{nl}_{#1} ({#2})}
\def\depth#1{\mathbf{d}(#1)}
\def\height#1{\mathsf{h}({#1})}
\newcommand{\bord}[1]{\mathsf{border}(#1)}

\def\cutstep#1#2{#1\mbox{-vs-}#2}
\def\cutelim{\to_{\cutr}}
\def\cutelims{\cutelim^*}

\def\cutelimnhl{\to_{\cutr+}}

\newcommand{\fapx}[1][\maximal]{\mathcal K({#1})}

\def\emptymset{[~]}
\def\mset#1{[#1]}
\newcommand{\fmset}[1]{\mathcal{M}_f({#1})}

\def\Sem#1{\left\{\!\!\left\{#1\right\}\!\!\right\}}
\def\sem#1{\{\!\{#1\}\!\}}

\def\fsem#1#2{\Sem{#1}_{#2}}
\def\fSem#1#2{\fsem{#1}{#2}}

\newcommand{\ptime}{\mathbf{P}}

\newcommand{\ppoly}{\ptime/\mathsf{poly}}

\newcommand{\cod}[1]{\underline{#1}}

\newcommand{\true}{\mathbf{1}}
\newcommand{\false}{\mathbf{0}}

\defedgetype{T}{draw, very thick, gray, opacity=.5}{}
\defedgetype{dT}{draw, very thick, gray, opacity=.5, densely dotted}{}

\newcommand{\ices}{\mathsf{ices}}
\newcommand{\mcices}{\mathsf{mc}\text{-}\mathsf{ices}}
\renewcommand{\lim}[2]{\mathsf{lim}_{#1}{\left( {#2 }\right)}}

\begin{document}

\maketitle

\begin{abstract}	
	We investigate non-wellfounded proof systems based on parsimonious logic, a weaker variant of linear logic where the exponential modality ! is interpreted as a constructor for streams over finite data.	Logical consistency is  maintained at a global level by adapting a standard progressing criterion.	
	We present an infinitary version of cut-elimination based on finite approximations, and  we prove that, in presence of the progressing criterion, it returns well-defined non-wellfounded proofs at its limit.
	Furthermore, we show that cut-elimination preserves the progressive criterion and various regularity conditions internalizing degrees of proof-theoretical uniformity.
	Finally, we provide a denotational semantics for our systems based on the relational model.
\end{abstract}


%

\section{Introduction}

\emph{Non-wellfounded proof theory} studies proofs as possibly infinite (but finitely branching) trees, where logical consistency is maintained via  global conditions called \emph{progressing} (or \emph{validity}) \emph{criteria}. In this setting, the so-called   \emph{regular} (also called \emph{circular}) proofs receive a special attention, as they  admit a finite description in terms of (possibly cyclic) directed graphs. 

This area of proof theory  makes its first appearance (in its modern guise) in  the modal $\mu$-calculus~\cite{niwinski1996games,dax2006proof}. Since then, it has been extensively investigated from many perspectives (see, e.g.,~\cite{brotherston2011sequent, Simpson17, Das2021,Kuperberg-Pous21}), establishing itself as an ideal setting for manipulating least and greatest fixed points, and hence for modeling induction and coinduction principles.  

Non-wellfounded proof theory has been applied to constructive fixed point logics i.e., with a computational interpretation based on the \emph{Curry-Howard correspondence}~\cite{Curry-Howard}. A key example can be found in the context  of \emph{linear logic} ($\mathsf{LL}$)~\cite{gir:ll}, a logic implementing  a finer control on resources thanks to the 	\emph{exponential} modalities $\oc$ and $\wn$. In this framework, the most extensively studied fixed point logic is  $\mu\mathsf{MALL}$, defined as the exponential-free fragment of $\mathsf{LL}$ with least and greatest fixed point operators (respectively,  $\mu$ and its dual  $\nu$)~\cite{BaeldeM07,BaeldeDS16}.

In~\cite{BaeldeM07} Baelde and Miller have shown that the exponentials can be recovered  in $\mu \mathsf{MALL}$ by exploiting the fixed points operators, i.e., by defining  $\oc A \dfn \nu X. (\mathbf{1} \with A \with (X \otimes X))$ and $\wn A\dfn \mu X. (\bot \oplus A \oplus (X \parr X))$.  
As these authors notice, the fixed point-based definition of $\oc$ and $\wn$  can be regarded as  a more permissive variant of the standard exponentials, since a proof of $\nu X. (\mathbf{1} \with A \with (X \otimes X))$ could be constructed using different proofs of $A$, whereas in $\mathsf{LL}$ a proof of $\oc A$  	is constructed uniformly using a single proof of $A$. 
This  proof-theoretical notion of  \emph{non-uniformity} is indeed  a central feature of the  fixed-point exponentials.

However, the above encoding is not free from issues.  First, as  discussed in full detail in~\cite{Farzad}, the encoding of the exponentials does not verify the Seely isomorphisms, 
syntactically expressed by the equivalence  $\oc (A \with B)\multimapboth (\oc A \otimes \oc B)$, 
an essential property for modeling exponentials in $\mathsf{LL}$.   Specifically, the  fixed-point definition of $\oc$  relies on the multiplicative connective  $\otimes$,  which forces an interpretation of $\oc A$ based on  lists rather than multisets. Secondly,  as pointed out   in~\cite{BaeldeM07}, there is a neat  mismatch between cut-elimination for the exponentials of $\mathsf{LL}$
and the one  for the fixed point exponentials of  $\mu\mathsf{MALL}$. While the first problem  is  related to   syntactic deficiencies  of the encoding, and  does not undermine further investigations on fixed point-based definitions of the exponential modalities, the second one is more critical.  These apparent differences between the two exponentials contribute to stressing an important aspect  in linear logic modalities, i.e., their  \emph{non-canonicity}~\cite{quatrini:phd,DanosJ03}\footnote{It is  possible to construct linear logic proof systems with alternative (non equivalent) exponential modalities (see,  e.g.,~\cite{NigamM09}).}.

On a parallel research thread, Mazza~\cite{Mazza15, MazzaT15} studied  \emph{parsimonious logic}, a variant of linear logic where the exponential modality $\oc$ satisfies Milner's law (i.e., $\oc A \multimapboth A \otimes \oc A$) and  invalidates   the implications $\oc A \multimap \oc \oc A$ (\emph{digging}) and $\oc A \multimap \oc A \otimes \oc A$ (\emph{contraction}).  In parsimonious logic, a proofs of  $\oc A$ can be interpreted  as a  \emph{stream} over (a finite set of)  proofs of  $A$, i.e., as a greatest fixed point, where  the   linear implications $A \otimes \oc A \multimap \oc A$ (\emph{co-absorption}) and $\oc A \multimap A \otimes \oc A$ (\emph{absorption}) can be computationally read as  
the \emph{push} and \emph{pop} 
operations on streams.  More specifically, a formula $\oc A$ is  introduced  by an \emph{infinitely branching rule} that takes a  finite set of proofs $\der_1, \ldots, \der_n$ of $A$ and  a (possibly non-recursive) function $f: \mathbb{N}\to \set{1, \ldots, n}$ as premises, and  constructs a proof of $\oc A$ representing a  stream of proofs of the form $\mathfrak{S}=\nwpstream{\der_{f(0)}}{\der_{f(1)}}{\der_{f(n)}}$.  Hence, parsimonious logic exponential modalities exploit in an essential way the above-mentioned  proof-theoretical non-uniformity, which in turn   deeply  interfaces with  notions of  non-uniformity from computational complexity~\cite{MazzaT15}. 

The analysis of parsimonious logic conducted in~\cite{Mazza15,MazzaT15} reveals that fixed point definitions of the exponentials are better behaving when  digging and  contraction  are discarded. On the other hand, the co-absorption rule cannot be derived in $\mathsf{LL}$, and so it prevents parsimonious logic becoming a genuine subsystem of the latter. This led the authors of the present paper to  introduce  \emph{parsimonious linear logic}, a co-absorption-free subsystem of linear logic that nonetheless allows a  stream-based interpretation of the exponentials.

We present two finitary proof systems for parsimonious linear logic: the system $\dpll$, supporting non-uniform exponentials, and $\pll$,  a fully uniform version. We investigate non-wellfounded  counterparts of $\dpll$ and $\pll$, adapting to our setting   the {progressing criterion} to maintain logical consistency. 
To recover the proof-theoretical behavior of $\dpll$ and $\pll$, we  identify further global conditions on non-wellfounded proofs, that is, some forms of regularity to capture the notions of uniformity and non-uniformity. 
This leads us to  two main non-wellfounded proof systems:  \emph{regular parsimonious linear logic} ($\cpll$), defined via the regularity condition and corresponding to $\pll$, and \emph{non-uniform parsimonious linear logic} ($\nupll$), defined via a \emph{weak regularity} condition and corresponding to $\dpll$.

The major contribution of this paper is the study of continuous cut-elimination in the setting of non-wellfounded parsimonious linear logic. We first  introduce  Scott-domains of  partially defined  non-wellfounded proofs,   ordered by an approximation relation. 
Then, we  define  special infinitary  proof rewriting  strategies called \emph{maximal and continuous infinitary cut-elimination strategies} ($\mcices$) which compute  (Scott-)continuous functions.  
Productivity  in this framework is established by showing that,  in presence of the progressing condition, these continuous functions return totally defined cut-free non-wellfounded proofs  (\Cref{thm:CCE}.\ref{MEGA:1}).
Moreover, we prove that they also preserve the (weak) progressing, the finite expandability, and the (weak) regularity conditions (\Cref{MEGA:2,MEGA:3,MEGA:4} in  \Cref{thm:CCE}). 

On a technical side, we stress that our methods and results distinguish   from previous approaches to cut-elimination in a non-wellfounded setting in many respects.  
First, 
we get rid of many technical notions typically introduced to prove infinitary cut-elimination, 
such as 
the \emph{multicut rule}
or the 
\emph{fairness conditions} (as in, e.g., \cite{fortier2013cuts,BaeldeDS16}), as these notions are subsumed by an \emph{approximation-by-approximation} approach to cut-elimination.
Furthermore, we prove productivity of cut-elimination and preservation  of progressiveness in a more direct and constructive way, i.e., without going through auxiliary proof systems and avoiding  arguments by contradiction (see, e.g., ~\cite{BaeldeDS16}). Finally,  we prove for the first time  preservation  of regularity properties under continuous cut-elimination, essentially exploiting  methods for compressing  transfinite rewriting sequences to $\omega$-long ones from~\cite{Terese, Saurin}.

Finally, we define a denotational semantics for non-wellfounded parsimonious logic based on  the relational model, with a standard multiset-based interpretation of the exponentials, and we show that this semantics is preserved under continuous cut-elimination (\Cref{thm:soundness-semantics}). 
We also prove that
extending non-wellfounded parsimonious linear logic with digging
prevents the existence of a cut-elimination result preserving the semantics (\Cref{thm:no-semantics-digging}).
Therefore, the impossibility of a stream-based definition of $\oc$ that validates digging (and contraction).

{For lack of space, proofs are in the appendix if omitted or sketched in the body of~the~paper.}

\section{Preliminary notions}\label{sec:prelim}

In this section we recall some basic notions from (non-wellfounded) proof theory, fixing the notation that will be adopted in this paper.

\subsection{Derivations and coderivations}\label{subsec:coderivations}

We assume that the  reader is familiar with the syntax of sequent calculus, e.g. \cite{troelstra_schwichtenberg_2000}. 
Here we specify some conventions adopted to simplify the content of  this paper.

In this work we consider (\defn{sequent}) \defn{rules} of the form
$\vlinf{\rrule}{}{\Gamma}{}$
or
$\vlinf{\rrule}{}{\Gamma}{\Gamma_1}$
or
$\vliinf{\rrule}{}{\Gamma}{\Gamma_1}{\Gamma_2}$,
and we refer 
to the sequents $\Gamma_1$ and $\Gamma_2$ as the \defn{premises}, 
and
to the sequent $\Gamma$ as the \defn{conclusion} of the rule $\rrule$.
To avoid technicalities of the sequents-as-lists presentation, we follow \cite{BaeldeDS16} and we consider \defn{sequents} as \emph{sets of occurrences of formulas} from a given set of formulas.
In particular, when we refer to a formula in a sequent we always consider a \emph{specific occurrence}~of~it.

\def\tree{\mathcal T}
\def\branch{\mathcal B}
\begin{definition}
	\label{def:coderivation}
	
	A (binary, possibly infinite) \defn{tree} $\tree$ is a subset of words in $\set{1,2}^*$
	that~contains the empty word $\epsilon$ (the \defn{root} of $\tree$)
	and is \emph{ordered-prefix-closed}
	(i.e., if $n \in \set{1,2}$ and $vn \in \mathcal{T}$, then $v \in \mathcal{T}$,
	and
	if moreover $v2 \in \mathcal{T}$, then $v1 \in \mathcal{T}$).
	Elements of a tree $\tree$ are called \defn{nodes}
	and 
	a node $vn\in\tree$ with $n\in\set{1,2}$ is 
	a \defn{child} of $v\in\tree$.
Given a tree $\tree$ and a node $v \in \tree$, a \defn{branch} $\branch$ of $\tree$ (from $v$) is a set of nodes in $\tree$ of the form $vw$ (for any $w \in \{1,2\}^*$) 
such that if they have at least one child in $\tree$ then they have exactly one child in $\branch$.

	A \defn{coderivation} over a set of rules $\Sys$
	is a labeling $\der$ of a tree by sequents such that
	if $v$ is a node with children $v_1,\ldots, v_n$ (with $n\in\set{0,1,2}$), 
	then there is an occurrence of a rule $\rrule$ in $\Sys$ with conclusion the sequent $\der(v)$ and premises the sequents $\der(v_1),\ldots, \der(v_n)$.
	The \defn{height} of $\rrule$ in $\der$ is the length of the node $v \in \set{1,2}^*$ such that $\der(v)$ is the conclusion of $\rrule$.

	The \defn{conclusion} of $\der$ is the sequent $\der(\epsilon)$.
	If $v$ is a node of the tree, the \defn{sub-coderivation} of $\der$ rooted at $v$ is the coderivation $\der_v$ defined by $\der_v(w)=\der(vw)$.

	A coderivation $\der$ is \defn{$\rrule$-free} (for a rule $\rrule \!\in\! \Sys$) if it contains no occurrence of $\rrule$.
	 It is \defn{regular} if it has finitely many distinct sub-coderivations; it is \defn{non-wellfounded} if it labels an infinite tree, and 
	it is a \defn{derivation} (with \defn{size} $\size{\der} \in \Nset$) if it labels a finite tree (with $\size{\der}$ nodes).
	
	Given a set of coderivations $\setDer$, a sequent $\Gamma$ is \defn{provable} in $\setDer$ (noted $\proves{\setDer} \Gamma$)
if there is a coderivation in $\setDer$ with conclusion $\Gamma$.
\end{definition}

While  derivations are usually represented as  finite trees, regular coderivations can be represented  as \emph{finite} directed (possibly cyclic) graphs: a cycle is created by linking the roots of two identical~subcoderivations.

\begin{definition}[Bar] 
	Let $\der$ be a coderivation.
	A set $\Nodes$ of nodes in $\der$ is a \defn{bar} (of $\der$) 
	if:
	\begin{itemize}
		\item any infinite branch of $\dD$ contains a node in $\Nodes$;
		\item any pair of nodes in $\Nodes$ are mutually incomparable (w.r.t.~the partial order in $\der$).
	\end{itemize}
	We say that a bar $\Nodes$ has \defn{height}  $h$ if every node in $\Nodes$ that is not a leaf of $\dD$ has height  $\geq h$.
\end{definition}

\section{Parsimonious Linear Logic}\label{sec:parsimoniousLogic}

In this paper we consider the set of  \defn{formulas}  for propositional multiplicative-exponential linear logic with units ($\mell$).
These are generated by a countable set of propositional variables $\atoms=\set{X,Y,\ldots}$ using the following grammar:
$$A,B \Coloneqq X \mid \cneg X   \mid A\ltens  B \mid A\lpar B \mid \oc A \mid \wn A \mid \lone \mid \lbot 
$$

A \defn{$\oc$-formula}  (resp.~\defn{$\wn$-formula}) is a formula of the form $\oc A$ (resp.~$\wn A$). 
\defn{Linear negation} $\cneg{(\cdot)}$ is defined by De Morgan's~laws
$\cneg{({\cneg A})}		= A$					 , 
$\cneg{(A \ltens B)} 	= \cneg A \lpar \cneg B$ , 
$\cneg {(\oc A)} 		= \wn \cneg A$			 , and 
$\cneg {(\lone)} 		= \lbot $
%
while \defn{linear implication} is defined as $A \limp B \coloneqq \cneg A \parr B$.

\begin{figure*}[!t]
	\begin{adjustbox}{max width=\textwidth}$
		\begin{array}{c}
			\vlinf{\axr}{}{A, \cneg A}{}
			\quad
			\vliinf{\cutr}{}{ {\Gamma}, {\Delta}}{ {\Gamma}, A}{\cneg A,{\Delta}}
			\quad 
			\vlinf{\lpar}{}{ {\Gamma}, A \lpar B}{ {\Gamma}, A , B}
			\quad 
			\vliinf{\ltens}{}{ {\Gamma}, {\Delta},A \ltens B}{ {\Gamma}, A}{B, {\Delta}}
			\quad
			\vlinf{\lone}{}{\lone}{}
			\quad
			\vlinf{\lbot}{}{ {\Gamma}, \lbot}{ {\Gamma}}
			\quad
			\vlinf{\fprule}{}{  \wn \Gamma,\oc A}{  \Gamma, A} 
			\quad 
			\vlinf{\wnwrule}{}{ {\Gamma}, \wn A}{ {\Gamma}}
			\quad 
			\vlinf{\wnbrule}{}{ {\Gamma}, {\wn A}}{ {\Gamma}, A, {\wn A}}
		\end{array}
	$\end{adjustbox}
	\caption{Sequent calculus rules of $\pll$.}
	\label{fig:sequent-system-pll}
\end{figure*}

\begin{definition}
	\defn{Parsimonious linear logic}, denoted by $\pll$, is the set of rules in
	\Cref{fig:sequent-system-pll}, that is,
	\defn{axiom} ($\axr$), 
	\defn{cut} ($\cutr$), 
	\defn{tensor} ($\ltens$), 
	\defn{par} ($\lpar$),
	\defn{one} ($\lone$),
	\defn{bottom} ($\lbot$),
	\defn{functorial promotion} ($\fprule$), 
	\defn{weakening} ($\wnwrule$),  
	\defn{absorption} ($\wnbrule$).
	Rules
	$\axr$, $\ltens$, $\lpar$, $\lone$ and $\lbot$
	are called \defn{multiplicative}, while rules $\fprule$, $\wnwrule$ and $\wnbrule$  are called \defn{exponential}.	
	We also denote by $\pll$ the set of derivations over the rules in $\pll$.
\end{definition}

\begin{figure*}[!t]
	\centering
	\adjustbox{max width=.9\textwidth}{
	$\begin{array}{c|c|c|c}
		\falseder
		&
		\trueder
		&
		\der_\mathsf{abs}
		&
		\der_\mathsf{der}
		\\\hline
		\vlderivation{
			\vlin{\lpar}{}{\wn (X \otimes \cneg{X\!}) \lpar \cneg{X\!} \lpar X}
			{
				\vlin{\lpar}{}{\wn (X \otimes \cneg{X\!}), \cneg{X\!} \lpar X}
				{
					\vlin{\wnwrule}{}{ \wn (X \otimes \cneg{X\!}), \cneg{X\!}\!, X}
					{
						\vlin{\axr}{}{\cneg{X\!}\!, X}{\vlhy{}}
					}
				}
			}
		}
		&
		%
		\vlderivation{
			\vliq{\lpar \times 2}{}{\wn (X \otimes \cneg{X\!}) \lpar \cneg{X\!} \lpar X} 
			{
					\vlin{\wnbrule}{}{\wn (X \otimes \cneg{X\!}),  \cneg{X\!}\!, X} 
					{
						\vlin{\wnwrule}{}{\wn (X \otimes \cneg{X\!}), X \otimes \cneg{X\!}\!, \cneg{X\!}\!, X} 
						{
							\vliin{\ltens}{}{X \otimes \cneg{X\!}\!, \cneg{X\!}\!, X} 
							{
								\vlin{\axr}{}{\cneg{X\!}\!, X}{\vlhy{}}
							}
							{
								\vlin{\axr}{}{\cneg{X\!}\!, X}{\vlhy{}}
							}
						}
					}
			}
		}
	&
		\vlderivation{
			\vlin{\lpar}{}{\wn \cneg{A} \lpar (A \ltens \oc A)}
			{
				\vlin{\wnbrule}{}{\wn \cneg{A}, A \ltens \oc A}
				{
					\vliin{\ltens}{}{\cneg{A}, \wn \cneg{A}, A \ltens \oc A}
					{
						\vlin{\axr}{}{\cneg{A}, A}{\vlhy{}}
					}
					{
						\vlin{\axr}{}{\wn \cneg{A}, \oc A}{\vlhy{}}
					}
				}
			}
		}
		& 
		%
		\vlderivation{
			\vlin{\lpar}{}{\wn \cneg{A} \lpar A}
			{
				\vlin{\wnbrule}{}{\wn \cneg{A}, A}
				{
					\vlin{\wnwrule}{}{\cneg{A}, \wn \cneg{A}, A}
					{
						\vlin{\axr}{}{\cneg{A}, A}{\vlhy{}}
					}
				}
			}
		}
	\end{array}
	$}
	\caption{Examples of derivations in $\pll$.}
	\label{fig:der:examples-of-derivations}
\end{figure*}

\begin{example}\label{ex:derivations}
	\Cref{fig:der:examples-of-derivations} gives some examples of derivation in $\pll$.
	The (distinct) derivations $\falseder$ and $\trueder$ prove the same formula $\Nat \dfn \oc(X \multimap X) \multimap X \multimap X$.
	The derivation $\dD_{\mathsf{abs}}$
	proves the 
	 \emph{absorption law} $\oc A \multimap A \ltens \oc A$; 
	the derivation $\dD_\mathsf{der}$
	proves
	the \emph{dereliction law} $\oc A \multimap A$.
\end{example}

The \defn{cut-elimination} relation $\cutelim$ in $\pll$
is the union of
\defn{principal} cut-elimination steps 
in  
\Cref{fig:cut-elim-finitary-lin} (\defn{multiplicative}) 
and
\Cref{fig:cut-elim-finitary-exp} (\defn{exponential})
and 
\defn{commutative} cut-elimination steps in \Cref{fig:cut-elim-finitary-comm}.																																																	
The reflexive-transitive closure of $\cutelim$ is noted $\cutelim^*$.

\begin{restatable}{theorem}{cutelimPLL}
	\label{thm:cut-elimination-pll} 
	
	For every $\der \in \pll$,
	there is a cut-free $\der'\in\pll$ such that $\der\cutelims\der'$.
\end{restatable}
\begin{proof}[Sketch of proof]
	We associate with any derivation $\der$ in $\pll$ 
	a derivation $\ptomp{\der}$
	in $\mell$ sequent calculus. 
	Thanks to additional commutative cut-elimination steps, 
	we prove that cut-elimination in $\mell$ rewrites $\ptomp{\der}$ to the translation of 
	a derivation in $\pll$.
	The termination of cut-elimination in $\pll$ follows from the result in $\mell$~\cite{pag:tor:StrongNorm}.
	Details are in \Cref{app:3}.
\end{proof}

Akin to light linear logic \cite{girard:98,lafont:soft,Roversi},  the exponential rules of $\pll$  are weaker than those in $\mell$: the usual promotion rule is replaced by $\fprule$ (\emph{functorial promotion}), and the usual contraction and  dereliction  rules  by $\wnbrule$. 
As a consequence, the  \emph{digging} formula $\oc A \limp \oc \oc A$
and the \emph{contraction} formula $\oc A \limp \oc A \otimes \oc A$ 
are not provable in $\pll$ (unlike the dereliction formula, \Cref{ex:derivations}).
This allows us to interpret computationally these weaker exponentials in terms~of~streams, as well as  to control the complexity of cut-elimination~\cite{Mazza15,MazzaT15}.

It is easy to show that $\mell = \pll + \text{digging}$: if we add the digging formula as an axiom (or equivalently, the \emph{digging rule} $\digr$ in \Cref{fig:digging}) to the set of rules in \Cref{fig:sequent-system-pll}, then the contraction formula becomes provable, and the obtained proof system coincides~with~$\mell$.

\begin{figure*}[t]
	\adjustbox{max width=\textwidth}{$\begin{array}{c}	
		\vlderivation{
			\vliin{\cutr}{}{\Gamma, A}{\vlin{\axr}{}{A, \cneg A}{\vlhy{}}}{\vlhy{\Gamma, A}}
		}
		\cutelim
		\vlderivation{\vlhy{\Gamma, A}}
%
		\qquad
		\vlderivation{
			\vliin{\cutr}{}{\Gamma, \Delta,\Sigma}{
				\vlin{\lpar}{}{\Gamma, A\lpar B}{\vlhy{\Gamma, A, B}}
			}{
				\vliin{\ltens}{}{ \Delta,\cneg A \ltens \cneg{B}\!, \Sigma}{\vlhy{ \Delta,\cneg A}}{\vlhy{\cneg{B}\!,\Sigma}}
			}
		}
		\cutelim
		\vlderivation{
			\vliin{\cutr}{}{\Gamma, \Delta, \Sigma}{
				\vliin{\cutr}{}{\Gamma, \Delta, B}{\vlhy{\Gamma, B,A}}{\vlhy{\cneg{A}\!, \Delta}}
			}{
				\vlhy{\cneg{B}\!,\Sigma}
			}
		}
		\qquad
		\vlderivation{
			\vliin{\cutr}{}{\Gamma}
			{
				\vlin{\lbot}{
				}{\Gamma, \lbot}{\vlhy{\Gamma}}
			}{
				\vlin{\lone}{}{\lone}{}
			}
		}
		\cutelim
		\vlderivation{
			\vlhy{\Gamma}
		}	
	\end{array}
	$}
	\caption{Multiplicative cut-elimination steps in $\pll$.}
	\label{fig:cut-elim-finitary-lin}
\end{figure*}

\begin{figure*}[t]
	\centering
	\adjustbox{max width=.9\textwidth}{$\begin{array}{c}
		\vlderivation{
			\vliin{\cutr}{}{\wn \Gamma, \wn \Delta, \oc B}{
				\vlin{\fprule}{}{\wn \Gamma, \oc A}{
					\vlhy{\Gamma, A}
				}
			}{
				\vlin{\fprule}{}{\wn \cneg {A}\!, \wn \Delta, \oc B}{
					\vlhy{\cneg {A}\!, \Delta, B}
				}
			}
		}
		\cutelim
		\vlderivation{
			\vlin{\fprule}{}{\wn \Gamma, \wn \Delta, \oc B}{
				\vliin{\cutr}{}{\Gamma, \Delta, B}{
					\vlhy{\Gamma, A}
				}{
					\vlhy{\cneg {A}\!, \Delta, B}
				}
			}
		}
		\qquad
		\vlderivation{
			\vliin{\cutr}{}{\wn \Gamma, \Delta}{
				\vlin{\fprule}{}{\wn \Gamma, \oc A}{
					\vlhy{\Gamma, A}
				}
			}{
				\vlin{\wnwrule}{}{\Delta,\wn \cneg A}{\vlhy{\Delta}}
			}
		}
		\cutelim
		\vlderivation{
			\vliq{\wnwrule}{}{\wn\Gamma,\Delta}{\vlhy{\Delta}}
		}
		\\
		\vlderivation{
			\vliin{\cutr}{}{\wn \Gamma, \Delta}{
				\vlin{\fprule}{}{\wn \Gamma, \oc A}{
					\vlhy{\Gamma, A}
				}
			}{
				\vlin{\wnbrule}{}{\Delta, \wn \cneg A}{
					\vlhy{\Delta,\cneg {A}\!, \wn\cneg A}
				}
			}
		}
		\cutelim
		\vlderivation{
			\vliq{\wnbrule}{}{\wn\Gamma, \Delta}{
				\vliin{\cutr}{}{\Gamma, \wn \Gamma,\Delta }{
					\vlhy{\Gamma, A}
				}{
					\vliin{\cutr}{}{\wn\Gamma,  \Delta, \cneg A}{
						\vlin{\fprule}{}{\wn \Gamma,\oc A}{\vlhy{\Gamma , A}}
					}{
						\vlhy{\Delta, \cneg {A}\!, \wn \cneg A}
					}
				}
			}
		}
	\end{array}
	$}
	\caption{Exponential cut-elimination steps in $\pll$.}
	\label{fig:cut-elim-finitary-exp}
\end{figure*}

\begin{figure*}[t]
	\adjustbox{max width=\textwidth}{$
		\begin{array}{c@{\quad\cutelim\;}c@{\qquad\qquad}c@{\quad\cutelim\;}c}
		\vlderivation{
			\vliin{\cutr}{}{\Gamma, \Delta}
			{\vlin{\rrule}{}{\Gamma, A}{\vlhy{  \Gamma_1, A}}}
			{\vlhy{\cneg{A}, \Delta}}
		}
		&
		\vlderivation{
			\vlin{\rrule}{}{\Gamma, \Delta} 
			{\vliin{\cutr}{}{\Gamma_1, \Delta}
				{\vlhy{   \Gamma_1, A }}
				{\vlhy{ \cneg{A}, \Delta}}
			}
		}
		&
		\vlderivation{
			\vliin{\cutr}{}{\Gamma, \Delta}
			{\vliin{\rrule}{}{\Gamma, A}{\vlhy{\Gamma_1,A}}{\vlhy{\Gamma_2}}}
			{\vlhy{\Delta, \cneg A}}
		}
		&
		\vlderivation{
			\vliin{\rrule}{}{\Gamma,\Delta} 
			{\vliin{\cutr}{}{\Gamma_1,\Delta}
				{\vlhy{\Gamma_1, A}}
				{\vlhy{\cneg{A},\Delta}}
			}
			{\vlhy{\Gamma_2}}
		}
	\end{array}
	$}
	\caption{Commutative cut-elimination steps in $\pll$, where $\rrule \neq \cutr$.}
	\label{fig:cut-elim-finitary-comm}
\end{figure*}

\section{Non-wellfounded Parsimonious Linear Logic }\label{sec:coder}

In linear logic, a formula $\oc A$ is interpreted as the availability of $A$ at will.
This  intuition still holds in $\pll$. Indeed, the Curry-Howard correspondence interprets rule $\fprule$ introducing the modality $\oc$ as an operator taking a derivation $\der$ of $A$ and creating a (infinite) \emph{stream} $\nwpstream \der \der \der$ of copies of the proof~$\der$.
Each element of the stream is accessed via the cut-elimination step $\fprule$ vs $\wnbrule$ in \Cref{fig:cut-elim-finitary-exp}: 
rule $\wnbrule$ is interpreted as an operator \emph{popping} one copy of $\der$ out of the stream.
Pushing these ideas further, Mazza \cite{Mazza15} introduced \emph{parsimonious logic} $\mathbf{PL}$, a type system (comprising rules $\fprule$ and $\wnbrule$) characterizing~the logspace decidable problems.

Mazza and Terui then introduced in~\cite{MazzaT15} another  type system,  $\mathbf{nuPL}_{\forall \ell}$, based on parsimonious logic 
and capturing the complexity class $\ppoly$ (i.e., the problems decidable by polynomial size families of boolean circuits~\cite{arora_barak_2009}). Their system is endowed with a \emph{non-uniform} version of the functorial promotion, which takes a  finite set of proofs $\der_1, \ldots, \der_n$ of $A$ and  a (possibly non-recursive) function $f \colon \mathbb{N}\to \set{1, \ldots, n}$ as premises, and  constructs a proof of $\oc A$ modeling the stream $\nwpstream{\der_{f(0)}}{\der_{f(1)}}{\der_{f(n)}}$.  This  typing rule  is the key tool  to  encode  the so-called  \emph{advices} for Turing machines,  an essential  step to show  completeness for $\ppoly$.

In a similar vein, we can endow $\pll$ with a non-uniform version of $\fprule$ called \defn{infinitely branching  promotion} ($\nuprule$), which constructs a 
stream $\nwpstream{\der_0}{\der_1}{\der_n}$ with finite support, i.e., made of \emph{finitely} many distinct derivations (of the same conclusion):\footnotemark
\footnotetext{Rule $\nuprule$ is reminiscent of the $\omega$-rule used in (first-order) Peano arithmetic to derive formulas of the form $\forall x\phi$ that cannot be proven in a uniform way.}

\begin{equation}\label{eqn:damiano-rule}
	\adjustbox{max width=.9\textwidth}{$\begin{array}{c|c}
		\vlupsmash{
			\vlderivation{
				\vliiiiin{\nuprule}{
					\text{\scriptsize $\set{\der_i\mid i\in\Nset}$ is finite}
				}{\wn \Gamma, \oc A}{\vldr{\der_{0}}{\Gamma,A}}{\vldr{\der_{1}}{\Gamma,A}}{\vlhy{\cdots}}{\vldr{\der_{n}}{\Gamma,A}}{\vlhy{\cdots}}
			}
		}
		&
		\vlinf{\ocwrule}{}{\oc A}{}
		\qquad
		\vliinf{\ocbrule}{}{\Gamma, \Delta, \oc A}{\Gamma, A}{\Delta, \oc A}
	\end{array}
	$}
\end{equation}

The side condition on $\nuprule$ provides a proof theoretic counterpart to the function $f \colon \mathbb{N}\to \set{1, \ldots, n}$  in $\mathbf{nuPL}_{\forall \ell}$. 
Clearly, $\fprule$ is subsumed by the rule $\nuprule$, as it corresponds to the special (uniform) case where $\der_i = \der_{i+1}$ for all $i \in \Nset$.

\begin{definition}
	\label{defn:dpll} 
	 We define the set of rules $\dpll \dfn \set{\axr,\ltens,\lpar,\lone,\lbot, \cutr,\wnbrule,\wnwrule,\nuprule}$.
	 We also denote by $\dpll$ the set of derivations over the rules in $\dpll$.\footnotemark
	\footnotetext{To be rigorous, this requires a slight change in \Cref{def:coderivation}: the tree labeled by a derivation in $\dpll$ must be over $\Nset^\omega$ instead of $\set{1,2}^*$, in order to deal with infinitely branching derivations.}
\end{definition}

There are some notable differences between $\dpll$ and Mazza and Terui's original system $~\mathbf{nuPL}_{\forall \ell}$~\cite{MazzaT15}.
As opposed to  $\dpll$,  $\mathbf{nuPL}_{\forall \ell}$ 
is formulated as an intuitionistic (type) system. Furthermore, to achieve  completeness for $\ppoly$,  these authors introduced  second-order quantifiers and  the  co-absorption ($\ocbrule$) and co-weakening ($\ocwrule$) rules displayed  in \eqref{eqn:damiano-rule}.

\emph{Cut-elimination} steps for $\dpll$ are 
in \Cref{fig:cut-elim-finitary-lin,fig:cut-elim-finitary-comm,fig:damiano-cut-elimination} (\Cref{fig:damiano-cut-elimination} is in \Cref{app:3} because we do not  use it: it just adapts the exponential steps to $\nuprule$).
In particular, the step $\cutstep{\nuprule}{\wnbrule}$ in \Cref{fig:damiano-cut-elimination} \emph{pops} the first premise $\der_0$ of $\nuprule$  out of the \mbox{stream~$\nwpstream{\der_0}{\der_1}{\der_n}$.}

\subsection{From infinitely branching proofs to non-wellfounded proofs}

In this paper we explore a dual approach to the one of $\mathbf{nuPL}_{\forall\ell}$ (and $\dpll$): 
instead of considering (wellfounded) derivations with infinite branching, we consider (non-wellfounded) coderivations with finite branching.
For this purpose, the infinitary rule $\nuprule$ of $\dpll$ is replaced by the binary rule below, called \defn{conditional promotion} ($\cprule$):
\begin{equation}\label{eqn:condProm}
	\hskip14em
	\vliinf{\cprule}{}{\wn \Gamma, \oc A}{\Gamma, A}{\wn \Gamma, \oc A}
\end{equation}

\begin{figure}
	$$
	\zeroder\coloneqq
	{\vlderivation{
			\vliin{\cutr}{}
			{\Gamma, A}
			{
				\vlin{\axr}{}{\cneg A, A}{\vlhy{}}
			}
			{
				\vliin{\cutr}{}
				{\Gamma,  A}
				{
					\vlin{\axr}{}{\cneg A, A}{\vlhy{}}
				}
				{\vlin{\cutr}{}{\Gamma, A}{\vlhy{\vdots}}}
			}
	}}
	\qquad 
	\wnder\coloneqq
	\vlsmash{\vlderivation{
		\vlin{\wnbrule}{}{\wn A}{
			\vlin{\wnbrule}{}{A, \wn A}{
				\vlin{\wnbrule}{}{A, A, \wn A}{\vlhy{\vdots}}
			}
		}
	}}
	$$
	\caption{
		Two {non-wellfounded} and non-\prog coderivations in $\nwpll$.
	}
\label{fig:nonProg}
\end{figure}

\begin{definition}
We define the set of rules $\nwpll \dfn \set{\axr,\ltens,\lpar, \lone, \lbot, \cutr,\wnbrule,\wnwrule,\cprule}$. 
We also denote by $\nwpll$ the set of coderivations over the rules in $\nwpll$.
\end{definition}
	
In other words, $\nwpll$ is the set of coderivations generated by the same rules as $\pll$, except that $\fprule$ is replaced by $\cprule$.
From now on, we will only consider coderivations in $\nwpll$.

\begin{example}\label{ex:coderivations}
	\Cref{fig:nonProg} shows two non-wellfounded coderivations in $\nwpll$:
	$\zeroder$ (resp. $\wnder$) has an infinite branch of $\cutr$ (resp.~$\wnbrule$) rules, and is (resp.~is not)~regular.
\end{example}

We can embed	$\pll$ and $\dpll$ into $\nwpll$ via the conclusion-preserving translations $\fprojp{(\cdot)} \colon \pll \to \nwpll$ and $\cprojp{(\cdot)} \colon \dpll \to \nwpll$ defined in \Cref{fig:translations-pll} by induction on derivations: they map all rules to themselves except $\fprule$ and $\nuprule$,  which are   ``unpacked'' into  non-wellfounded coderivations that  iterate infinitely many times the rule $\cprule$.

\begin{figure*}[t]
	\adjustbox{max width=\textwidth}{$\begin{array}{l}
		\begin{array}{c|c}
			\fproj{
				\vlderivation{
					\vlin{\rrule}{}
					{\Gamma}
					{\vldr{\der}{\Gamma'}}
				}
			}
			\dfn
			\vlderivation{
				\vlin{\rrule}{}
				{\Gamma}
				{\vldr{\fprojp{\der}}{\Gamma'}}
			}
			\qquad
			\fproj{
				\vlderivation{
					\vliin{\trule}{}{\Gamma}{
						\vldr{\der_1}{\Gamma_1}
					}{
						\vldr{\der_2}{\Gamma_2}
					}
				}
			}
			\dfn
			\vlderivation{
				\vliin{\trule}{}{\Gamma}{
					\vldr{\fprojp{\der_1}}{\Gamma_1}
				}{
					\vldr{\fprojp{\der_2}}{\Gamma_2}
				}
			}
		&
			\fproj{
				\vlderivation{
					\vlin{\fprule}{}
					{\wn \Gamma, \oc A}
					{\vldr{\der}{\Gamma, A}}
				}
			}
			\dfn
			\vlderivation{
				\vliin{\cprule}{}
				{\wn \Gamma, \oc A}
				{\vldr{\fprojp{\der}}{ \Gamma, A}}
				{
					\vliin{\cprule}
					{}
					{\wn \Gamma, \oc A}
					{\vldr{\fprojp{\der}}{ \Gamma, A}}
					{\vlin{\cprule}{}{\wn \Gamma, \oc A}{\vlhy{\vdots}}}
				}
			}
		\\
			\cproj{
				\vlderivation{
					\vlin{\rrule}{}
					{\Gamma}
					{\vldr{\der}{\Gamma'}}
				}
			}
			\dfn
			\vlderivation{
				\vlin{\rrule}{}
				{\Gamma}
				{\vldr{\cprojp{\der}}{\Gamma'}}
			}
			\qquad
			\cproj{
				\vlderivation{
					\vliin{\trule}{}{\Gamma}{
						\vldr{\der_1}{\Gamma_1}
					}{
						\vldr{\der_2}{\Gamma_2}
					}
				}
			}
			\dfn
			\vlderivation{
				\vliin{\trule}{}{\Gamma}{
					\vldr{\cprojp{\der_1}}{\Gamma_1}
				}{
					\vldr{\cprojp{\der_2}}{\Gamma_2}
				}
			}
		&
			\cproj{
				\vlderivation{
					\vliiiin{\nuprule}{}{\wn \Gamma, \oc A}{
						\vldr{\der_{0}}{\Gamma,A}
					}{
						\vlhy{\cdots}
					}{
						\vldr{\der_{n}}{\Gamma,A}}{\vlhy{\cdots}
					}
				}
			}
			\dfn
			\vlderivation{
				\vliin{\cprule}{}{\wn \Gamma, \oc A}{
					\vldr{\der_0^\bullet}{\Gamma, A}
				}{
					\vliin{\cprule}{}{
							\reflectbox{$\ddots$}
						}{
							\vldr{\der_n^\bullet}{\Gamma, A}
						}{
							\vlin{\cprule}{}
							{{\wn \Gamma, \oc A}}
							{\vlhy{\vdots}}
						}
				}
			}
		\end{array}
		\\
		\mbox{for all }
		\rrule\in\set{\lpar,\lbot,\wnwrule,\wnbrule}
		\mbox{ and }
		\trule\in\set{\cutr,\ltens} 
		\mbox{ ($\axr$ and $\lone$ are translated by themselves).}
	\end{array}$}
	\caption{
		Translations $\fprojp{(\cdot)}$ from $\pll$ to $\nwpll$, and $\cprojp{(\cdot)}$ from  $\dpll$ to $\nwpll$.
	}
	\label{fig:translations-pll}
\end{figure*}

An infinite chain of $\cprule$ rules (\Cref{eq:box}) is a structure of interest in itself in~$\nwpll$.

\begin{figure}[t]
	\centering
	\vspace{-15pt}
	\adjustbox{max width=.8\textwidth}{$
		\der = 
		\derstream{\der_0,\ldots, \der_n} 
		=
		\vlderivation{
			\vliin{\cprule}{}{\wn \Gamma, \oc A}{\vldr{\der_0}{\Gamma, A}}{
				\vliin{\cprule}{}{{\wn \Gamma, \oc A}}{\vldr{\der_1}{\Gamma, A}}{				 			
					\vliin{\cprule}{}{\reflectbox{$\ddots$}}{\vldr{\der_n}{\Gamma, A}}{
						\vlin{\cprule}{}{{\wn \Gamma, \oc A}}{\vlhy{\vdots}}
					}
				}
			}
		}
	$}
	\caption{A non-wellfounded box in $\nwpll$.}
	\label{eq:box}
\end{figure}

\begin{definition}\label{defn:boxes} 
	A   \defn{non-wellfounded box} 
	(\nwbox for short)
	is a coderivation $\der$
	with an infinite branch $\set{\epsilon, 2,22,\dots}$
	(the  \defn{main branch} of $\der$) all labeled by $\cprule$ rules as in~\Cref{eq:box}, 
	where $\oc A$ in the conclusion is the \defn{principal formula} of $\der$, and	$\der_0, \der_1, \ldots$ are the \defn{calls} of $\der$. 
	We denote $\der$ by $\derstream{\der_0,\ldots, \der_n}$.
	
	Let $\nwpromotion=\derstream{\der_0,\ldots, \der_n}$ be a $\nwbox$.
	We may write $\nwpromotion(i)$ to denote $\der_i$.
	We say that $\nwpromotion$ has \defn{finite support} ({resp.~}is \defn{periodic} with \defn{period} $k$)
	if $\set{\nwpromotion(i)\mid i\in\Nset}$ is finite ({resp.~}if $\nwpromotion(i)= \nwpromotion(k+i)$ for any $i\in\Nset$).
	A coderivation $\der$ 
	has  \defn{finite support} ({resp.~}is \defn{periodic}) 
	if any $\nwbox$ in $\der$ has finite support ({resp.~}is periodic).
\end{definition}

\begin{example}\label{ex:nonReg}
	Consider the following \nwbox of the formula $\oc \Nat$, where $\Nat \dfn \oc(X \multimap X) \multimap X \multimap X$ has at two distinct derivations $\falseder$ and $\trueder$ (\Cref{ex:derivations}), and $i_j \in \{\false,\true\}$ for all $j \in \Nset$.
	\begin{equation}\label{eq:stramCod}
		\hfill
		\adjustbox{max width=.55\textwidth}{$
		\derstream{\cod{i_0},\ldots,\cod{i_n}}
		\quad =\qquad
		\vlderivation{
			\vliin{\cprule}{}{\oc \Nat}{
				\vldr[16]{\cod{i_0}}{\Nat}
			}{
				\vliin{\cprule}{}{{ \oc \Nat}}{
					\vldr[16]{\cod{i_1}}{\Nat}
				}{
					\vliin{\cprule}{}{
						\reflectbox{$\ddots$}
					}{
						\vldr[16]{\cod{i_n}}{\Nat}
					}{
						\vlin{\cprule}{}{{\oc \Nat}}{\vlhy{\vdots}}
					}
				}
			}
		}$}
		\hfill
	\end{equation}
	Thus $\derstream{\cod{i_0},\ldots,\cod{i_n}}$ has finite support, as its only calls can be $\falseder$ or $\trueder$,
	and it is periodic if and only if so is the infinite sequence $\seq{i_0,\ldots, i_n,\ldots} \in \{\false,\true\}^\omega$.
\end{example}

The \emph{cut-elimination} steps $\cutelim$ for $\nwpll$ are in \Cref{fig:cut-elim-finitary-lin,fig:cut-elim-finitary-comm,fig:cut-elim-pll}.
Computationally, they allow the $\cprule$ rule to be interpreted as a \emph{coinductive} definition of a stream of type $\oc A$ from a stream of the same type to which an element of type $A$ is prepended.
In particular, the cut-elimination step $\cprule$ vs $\wnbrule$ accesses the head of a stream: 
rule $\wnbrule$ acts as a \emph{popping}~operator.

As a consequence, the $\nwbox$ in \Cref{eq:box} constructs a stream $\nwpstream{\der_0}{\der_1}{\der_n}$ similarly to $\nuprule$ but, unlike the latter, all the $\der_i$'s may be pairwise distinct.
The reader {expert in} linear logic can see a $\nwbox$ as a box with possibly \emph{infinitely many} distinct contents (its calls), while usual linear logic boxes (and $\fprule$ in $\pll$) provide infinitely many copies of the~\emph{same}~content.

Rules $\fprule$ in $\pll$ and $\nuprule$ in $\dpll$ are mapped by $\fprojp{(\cdot)}$ and $\cprojp{(\cdot)}$ into \nwboxes, which are non-wellfounded coderivations. 
Hence, the cut-elimination steps $\fprule$ vs $\fprule$ in $\pll$ and $\nuprule$ vs $\nuprule$ in $\dpll$ can  only be simulated by infinitely many cut-elimination steps in $\nwpll$.

Note that $\zeroder \in \nwpll$ in \Cref{fig:nonProg} is not $\cutr$-free, and if $\zeroder \cutelim \der$ then $\der = \zeroder$: thus~$\zeroder$  cannot reduce to a $\cutr$-free coderivation, and so the cut-elimination theorem fails~in~$\nwpll$.

\begin{figure*}[t]
	\adjustbox{max width=\textwidth}{$\begin{array}{c}
			\vlderivation{
				\vliin{\cutr}{}{\wn \Gamma, \wn \Delta, \oc B}{
					\vliin{\cprule}{}{\wn \Gamma, \oc A}{
						\vlhy{\Gamma, A}
					}{
						\vlhy{\wn \Gamma, \oc A}
					}
				}{
					\vliin{\cprule}{}{\wn \cneg A\!, \wn \Delta, \oc B}{
						\vlhy{\cneg{A}\!, \Delta, B}
					}{
						\vlhy{\wn \cneg{A}\!, \wn \Delta, \oc B}
					}
				}
			}
			\cutelim
			\vlderivation{
				\vliin{\cprule}{}{\wn \Gamma, \wn \Delta, \oc B}{
					\vliin{\cutr}{}{\Gamma, \Delta, B}{
						\vlhy{\Gamma, A}
					}{
						\vlhy{\cneg A, \Delta, B}
					}
				}{
					\vliin{\cutr}{}{\wn\Gamma, \wn\Delta,\oc B}{
						\vlhy{\wn\Gamma, \oc A}
					}{
						\vlhy{\wn \cneg{A}\!, \wn \Delta, \oc B}
					}
				}
			}
			\\\\
			\vlderivation{
				\vliin{\cutr}{}{\wn \Gamma, \Delta}{
					\vliin{\cprule}{}{\wn \Gamma, \oc A}{
						\vlhy{\Gamma, A}
					}{
						\vlhy{\wn \Gamma, \oc A}
					}
				}{
					\vlin{\wnwrule}{}{\Delta,\wn \cneg A}{\vlhy{\Delta}}
				}
			}
			\cutelim
			\vlderivation{
				\vliq{\wnwrule}{}{\wn\Gamma,\Delta}{\vlhy{\Delta}}
			}
			\qquad
			\vlderivation{
				\vliin{\cutr}{}{\wn \Gamma, \Delta}{
					\vliin{\cprule}{}{\wn \Gamma, \oc A}{
						\vlhy{\Gamma, A}
					}{
						\vlhy{\wn \Gamma, \oc A}
					}
				}{
					\vlin{\wnbrule}{}{\Delta, \wn \cneg A}{
						\vlhy{\Delta,\cneg{A}\!, \wn\cneg A}
					}
				}
			}
			\cutelim
			\vldownsmash{\vlderivation{
				\vliq{\wnbrule}{}{\wn\Gamma, \Delta}{
					\vliin{\cutr}{}{\Gamma, \wn \Gamma,\Delta }{
						\vlhy{\Gamma, A}
					}{
						\vliin{\cutr}{}{\wn\Gamma,  \Delta, \cneg A}{
							\vlhy{\wn \Gamma,\oc A}
						}{
							\vlhy{\Delta, \cneg{A}\!, \wn \cneg A}
						}
					}
				}
			}}
		\end{array}$}
	
	\caption{Exponential cut-elimination steps for coderivations of $\nwpll$.}
	\label{fig:cut-elim-pll}
\end{figure*}

\subsection{Consistency via a \progness}

In a non-wellfounded setting such as $\nwpll\!$, any sequent is provable.
Indeed, the (non-wellfounded) coderivation $\zeroder$ in \Cref{fig:nonProg} shows that any non-empty sequent (in particular, any formula) is provable in $\nwpll$,
and the empty sequent is provable in $\nwpll$ by applying the $\cutr$ rule on the conclusions $B$ and $\cneg B$ (for any formula $B$) of two  derivations $\zeroder$.

The standard way to recover logical consistency in non-wellfounded proof theory is to introduce a global soundness condition on coderivations, called \defn{\prog criterion}. 
In $\nwpll$, this criterion relies on tracking occurrences of $\oc$-formulas in a coderivation.

\begin{definition}\label{defn:octhread}
	Let $\dD$ be a coderivation in $\nwpll$.
	It is \defn{weakly \prog} if every infinite branch contains infinitely many right premises of $\cprule$-rules.
	
	An occurrence of formula in a premise of a rule $\rrule$ is the \defn{parent} of an occurrence of a formula in the conclusion if they are connected according to the edges depicted in \Cref{fig:threads}.

	A \defn{\octhread} (resp.~\defn{\wnthread}) in $\der$
	is a maximal sequence $(A_i)_{i \in I}$ of $\oc$-formulas (resp.~$\wn$-formulas)
	for some downward-closed $I \subseteq \Nset$ such that $A_{i+1}$ is the parent of $A_i$ for all $i \in I$.
	A \octhread $(A_i)_{i \in I}$ is \defn{\prog} if $A_j$ is in the conclusion of a $\cprule$ for infinitely many $j\in I$.

	$\der$ is \defn{\prog} if every infinite branch contains a \prog \octhread. 
		We define $\ppll$ (resp.~$\wppll$) as the set of \prog (resp.~\wprog) coderivations~in~$\nwpll$.
\end{definition}

\begin{figure*}[!t]
	\begin{adjustbox}{max width=\textwidth}$
			\begin{array}{c}
				\vlinf{\axr}{}{A, \cneg A}{}
				\quad
				\vliinf{\cutr}{}{
					\vF1_1, \ldots,\vF2_n,\vG1_1,\ldots,\vG2_m
				}{
					\vF3_1,\dots\vF4_n, A
				}{\cneg{A},\vG3_1,\ldots,\vG4_m}
				\Tedges{F1.center/F3.center,F2.center/F4.center,G1.center/G3.center,G2.center/G4.center}
				\quad 
				\vlinf{\lpar}{}{
					\vF1_1, \ldots,\vF2_n,\vA2\lpar \vB2
				}{
					\vF3_1,\dots\vF4_n, \vA1\;,\;\vB1
				}
				\Tedges{F1.center/F3.center,F2.center/F4.center,A1.center/A2.center,B1.center/B2.center}
				\quad 
				\vliinf{\ltens}{}{
					\vF1_1, \ldots,\vF2_n,\vA2\ltens \vB2,\vG1_1,\ldots,\vG2_m,
				}{
					\vF3_1,\dots\vF4_n, \vA1
				}{\vB1,\vG3_1,\ldots,\vG4_m}
				\Tedges{F1.center/F3.center,F2.center/F4.center,G1.center/G3.center,G2.center/G4.center,A1.center/A2.center,B1.center/B2.center}
				\\
				\vlinf{\lone}{}{\lone}{}
				\quad
				\vlinf{\lbot}{}{\vF1_1, \ldots,\vF2_n, \lbot}{ {\vF3_1,\dots,\vF4_n \quad}}
				\Tedges{F1.center/F3.center,F2.center/F4.center}
				\quad
				\vliinf{\cprule}{}{
					\wn \vF1_1, \ldots,\wn \vF2_n,\oc \vA1
				}{F_1,\ldots, F_n, A\quad}{
					\wn \vF3_1, \ldots,\wn\vF4_n,\oc \vA2
				}
				\Tedges{F1.center/F3.center,F2.center/F4.center,A1.center/A2.center}
				\quad 
				\vlinf{\wnwrule}{}{\vF1_1, \ldots,\vF2_n, \wn A}{ {\vF3_1,\dots,\vF4_n \quad}}
				\Tedges{F1.center/F3.center,F2.center/F4.center}
				\quad 
				\vlinf{\wnbrule}{}{\vF1_1, \ldots,\vF2_n, \wn \vA2}{ {\vF3_1,\dots,\vF4_n,A, \wn \vA1 }}
				\Tedges{F1.center/F3.center,F2.center/F4.center,A1.center/A2.center}
			\end{array}
		$\end{adjustbox}
	\caption{$\nwpll$ rules: edges connect a formula in the conclusion with its parent(s) in a premise.}
	\label{fig:threads}
\end{figure*}

\begin{remark}\label{rem:prog}
	Clearly, any \prog coderivation is weakly \prog too, but the converse fails (\Cref{ex:prog}), therefore $\ppll \subsetneq \wppll$.
	Moreover, the main branch of {any} \nwbox contains by definition a \prog \octhread of its principal formula.
\end{remark}

\newemptyvertex{ghost}{}
\begin{example}\label{ex:prog}
	Coderivations in \Cref{fig:nonProg} are not weakly \prog (hence, not \prog):
	the rightmost branch of $\zeroder$, i.e., the branch  $\set{\epsilon,2,22,\ldots}$,
	and the unique branch of $\wnder$
	are infinite and contain no $\cprule$-rules.
	In contrast, the \nwbox  $\derstream{\cod{i_0},\ldots,\cod{i_n}}$ in \Cref{ex:nonReg} is \prog by \Cref{rem:prog}, since its main branch is the only infinite branch.
	Below, a regular, weakly \prog but not \prog coderivation ($\oc X$ in the conclusion of $\cprule$ is a cut formula, so the branch $\set{\epsilon, 2, 21, 212, 2121, \dots}$ is infinite but has no progressing \octhread).
	$$
	\small
		\vlderivation{
			\vliin{\cprule}{}{\wn\vnX1 \;, \;\;\;\oc\vX1}{
				\vlin{\axr}{}{X, \ \cneg{X}}{\vlhy{}}}{
					\vliin{\cutr}{}{\wn \vnX2 \;\;,  \;\;\;\; \oc\vX2}{
						\vliin{\cprule}{}{\wn\vnX3, \ \rclr{\oc\vX3}}{
							\vlin{\axr}{}{X, \ \cneg{X}}{\vlhy{}}}{
							\vliin{\cutr}{}{\wn \vnX4, \ \oc\vX4}{
								\vlin{\cprule}{}{\wn\vnX5, \rclr{\oc \vX5}}{
									\vlhy{  
										\vghost2{~}\;\; \vdots \quad\vghost1{~}
									}
								}
							}{
								\vlin{\axr}{}{\rclr{\wn \cneg X},\oc \vX7}{\vlhy{}}
							}
						}
					}{
						\vlin{\axr}{}{\rclr{\wn \cneg X},\oc \vX6}{\vlhy{}}
					}
			}
		}
		\Tedges{
			X1.225/X1.135,
			X1.135/X2.225,
			X2.225/X2.135,
			X2.135/X6.225,
			X6.225/X6.135,
			X3.225/X3.135,
			X3.135/X4.225,
			X4.225/X4.135,
			X4.135/X7.225,
			X7.225/X7.135,
			X5.225/ghost1.south%
		}
		\dTedges{ghost1.south/ghost1.north}
		\Tedges{
			nX1.225/nX1.135,
			nX1.135/nX2.225,
			nX2.225/nX2.135,
			nX2.135/nX3.225,
			nX3.225/nX3.135,
			nX3.135/nX4.225,
			nX4.225/nX4.135,
			nX4.135/nX5.225,
			nX5.225/nX5.135,
			nX5.135/ghost2.south%
		}
		\dTedges{ghost2.south/ghost2.north}
	$$
\end{example}

\begin{lemma}\label{lem:weak-prog}
	Let $\Gamma$ be a sequent.
	Then, $\proves{\pll} \Gamma$ if and only if $\proves{\wppll} \Gamma$.
\end{lemma}
\begin{proof}
	Given $\dD\in\pll$, $\cprojp{\dD}\in \nwpll$ preserves the conclusion and is progressing, hence weakly \prog (see \Cref{rem:prog}).
	Conversely, given a weakly \prog coderivation $\dD$, 
	we define 
	a derivation $\dD^f\in\pll$ with the same conclusion
	by applying, bottom-up,~the~translation:
	\begin{equation*}
	\adjustbox{max width=.9\textwidth}{$
	\finproj{
		\vlderivation{
			\vlin{\rrule}{}
			{\Gamma}
			{\vldr{\der}{\Gamma'}}
		}
	}
	\dfn
	\vlderivation{
		\vlin{\rrule}{}
		{\Gamma}
		{\vldr{\der^f}{\Gamma'}}
	}
	\quad
	\finproj{
		\vlderivation{
			\vliin{\rrule}{}{\Gamma}{
				\vldr{\der_1}{\Gamma_1}
			}{
				\vldr{\der_2}{\Gamma_2}
			}
		}
	}
	\dfn
	\vlderivation{
		\vliin{\rrule}{}{\Gamma}{
			\vldr{\finprojp{\der_1}}{\Gamma_1}
		}{
			\vldr{\finprojp{\der_2}}{\Gamma_2}
		}
	}
	\qquad
	\finproj{
		\vlderivation{
			\vliin{\cprule}{}{\wn \Gamma, \oc A}{
				\vldr{\der}{\Gamma, A}
			}{
				\vldr{\der'}{\wn \Gamma, \oc A}
			}
		}
	}
	\dfn
	\vlderivation{
		\vlin{\fprule}{}{
			\wn\Gamma , \oc A
		}{
			\vldr{\der^f}{\Gamma, A}
		}
	}
	$}
	\end{equation*}
	with $\rrule \neq \cprule$.
	Note that the derivation $\dD^f$ is well-defined because $\dD$ is weakly \prog.
\end{proof}

\begin{corollary}
	The empty sequent is not provable in $\wppll$ (and hence in $\ppll$).
\end{corollary}
\begin{proof}
	If the empty sequent were provable in $\wppll$, then there would be a cut-free derivation $\der \in \pll$ of the empty sequent by \Cref{lem:weak-prog,thm:cut-elimination-pll}, but this is impossible since $\cutr$ is the only rule in $\pll$ that could have the empty sequent in its conclusion.
\end{proof}

\subsection{Recovering (weak forms of) regularity}

The \progness cannot capture the finiteness condition of the rule $\nuprule$ in the derivations in $\dpll$.
By means of example, consider the $\nwbox$ 
below, which is progressive but cannot be the image of 
the rule $\nuprule$ via $\cprojp{(\cdot)}$ (see \Cref{fig:translations-pll}) since $\set{\der_i\mid i\in \Nset}$ is infinite.
\vspace{-15pt}
\begin{equation}\label{eqn:non-finite-support}
	\small
	\vlderivation{
		\vliin{\cprule}{}{\oc\oc\Nat}{\vldr{\der_0}{\oc\Nat}}{
			\vliin{\cprule}{}{{\oc\oc\Nat}}{\vldr{\der_1}{\oc \Nat}}{				 			
				\vliin{\cprule}{}{\reflectbox{$\ddots$}}{\vldr{\der_n}{\oc\Nat}}{
					\vlin{\cprule}{}{{\oc\oc\Nat}}{\vlhy{\vdots}}
				}
			}
		}
	}
	\qquad
	\mbox{\normalsize with
		$\der_i=\derstream{\underbrace{{\scriptstyle{\trueder,\ldots,\trueder}}}_i,\falseder}$
		for each $i\in\Nset$.
	}
\end{equation}	

To identify in $\ppll$ the coderivations corresponding to derivations in $\dpll$ and in $\pll$ via the translations $\cprojp{(\cdot)}$ and $\fprojp{(\cdot)}$, respectively, we need additional conditions.

\begin{definition}\label{def:wR}
	\label{defn:weak-regulairity}  

	A coderivation is 
	\defn{weakly regular} if it has only finitely many distinct sub-coderivations 
	whose conclusions are left premises of $\cprule$-rules;
	it is \defn{\parsimonious} if any branch contains finitely many $\cutr$ and $\wnbrule$ rules.
	We denote by $\nupll$ (resp.~$\cpll$) the set of weakly regular (resp.~regular) and \parsimonious coderivations in {$\ppll$}.
\end{definition}

\begin{remark}\label{prop:weak-regular-finite-support}
	Regularity implies weak regularity and the converse fails as 
	shown in \Cref{ex:weak-regular} below, therefore $\cpll \subsetneq \nupll$.	
		Moreover,  $\der {\in \nwpll}$ is  regular (resp.~weakly regular) 
		if and only if  any $\nwbox$ in $\der$ is periodic (resp.~has finite support).
\end{remark}

\begin{example}\label{ex:weak-regular}
	{Coderivations $\zeroder$ and $\wnder$ in \Cref{fig:nonProg}} are not \parsimonious, as their infinite branch has infinitely many $\cutr$ or  $\wnbrule$,
	but they are weakly regular, since they have no $\cprule$ rules.
	The coderivation in \eqref{eqn:non-finite-support} is not weakly regular because $\set{\der_i\mid i\in \Nset}$ is infinite.
	
	An example of a weakly regular but not regular coderivation is  the  $\nwbox$ 
	$\derstream{\cod{i_0},\ldots,\cod{i_n}}$ in \Cref{ex:nonReg} when the infinite sequence 
	$(i_j)_{j \in \Nset} \in \set{\false,\true}^\omega$ is not periodic: in each rule $\cprule$ there, the left premise can only be $\falseder$ or $\trueder$ (so the $\nwbox$ is weakly regular), but the right premise is a distinct coderivation (so the $\nwbox$ is not regular). 
	Moreover, that $\nwbox$ is \parsimonious since it contains no $\wnbrule$ or $\cutr$.
\end{example}

The sets $\cpll$ and $\nupll$ are the non-wellfounded counterparts of $\pll$ and $\dpll$, respectively.
Indeed, we have the following correspondence via the translations $\fprojp{(\cdot)}$ and $\cprojp{(\cdot)}$. 

\begin{restatable}{proposition}{translations}
	\label{prop:translations}
	\begin{enumerate}
		\item\label{p:translations-foward}	If $\der \in \pll$ (resp. $\der \in \dpll$) with conclusion $\Gamma$, then $\fprojp{\der} \in \cpll$ (resp. $\cprojp{\der} \!\in \nupll$) with conclusion $\Gamma$, and every $\cprule$ in $\fprojp{\der}$ (resp. $\cprojp{\der}$) belongs to a~$\nwbox$.
		\item\label{p:translations-backward} If $\der' \in \cpll$ (resp. $\der' \in \nupll$) and every $\cprule$ in $\der'$ belongs to a $\nwbox$, then there is $\der \in \pll$ (resp. $\der \in \dpll$) such that $\fprojp{\der} = \der'$ (resp. $\cprojp{\der} = \der'$).  
	\end{enumerate}

\end{restatable}

\Prog and weak \prog coincide in finite expandable coderivations.

\begin{restatable}{lemma}{infinitebrances}
	\label{lem:infinite-branches}\label{cor:simple-structure}
	Let $\der \in \nwpll$ be  \parsimonious.  If $\der\in\wppll$ then any infinite branch contains the principal branch of a $\nwbox$.	Moreover, $\der \in \ppll$  iff $\der \in \wppll$.
\end{restatable}
\begin{proof}
	Let $\der\in\wppll$ be \parsimonious, and let $\mathcal{B}$ be an infinite branch in $\der$. By \parsimony there is  $h\in\Nset$ such that $\mathcal{B}$ contains no conclusion of a $\cutr$ or $\wnbrule$ with height greater than $h$. 
	Moreover, by weakly \prog there is an infinite sequence $h\leq h_0 < h_1 <  \ldots <h_n <\ldots$ such that the sequent of $\mathcal{B}$ at height $h_i$ has shape $\wn \Gamma_i, \oc A_i$.
	By inspecting the rules in~\Cref{fig:sequent-system-pll},  each such $\wn \Gamma_i, \oc A_i$ can be either the conclusion of either a $\wnwrule$ or a $\cprule$ (with right premise $\wn \Gamma_i, \oc A_i$). 
	So, there is a $k$ large enough such that, for any $i \geq k$, only the latter case applies (and, in particular, $\Gamma_i=\Gamma$ and $A_i=A$ for some $\Gamma, A$).
	Therefore, $h_k$ is the root of a  \nwbox. 
	This also shows $\der \in \ppll$. By \Cref{rem:prog}, $\ppll \subseteq \wppll$.
\end{proof}

By inspecting the steps in \Cref{fig:cut-elim-finitary-comm,fig:cut-elim-finitary-lin,fig:cut-elim-pll}, we prove the following preservations.
\begin{restatable}{proposition}{preserves}
	\label{prop:cut-elim-preserves-finexp-reg-weakreg}
	Cut elimination preserves weak-regularity, regularity and \parsimony.
	Therefore, if $\dD\in\sysX$ with $\sysX\in\set{\cpll,\nupll}$ and
	$\der \cutelim \der'$, then also $\der' \in \sysX$.
\end{restatable}

\section{Continuous cut-elimination}\label{sec:ccutelim}


Cut-elimination for (finitary) sequent calculi   proceeds by introducing   a  proof rewriting strategy that stepwise decreases an appropriate  {termination ordering (see, e.g,~\cite{troelstra_schwichtenberg_2000}). Typically, these proof rewriting strategies  consist on pushing upward    the topmost  cuts via the cut-elimination steps in order to eventually eliminate them. 
	
A somewhat dual approach is investigated in the context of non-wellfounded proofs~\cite{BaeldeDS16,FortierS13}. It consists on \emph{infinitary} proof rewriting strategies that gradually push upward  the  bottommost cuts. In this setting, the progressing condition is essential to guarantee  \emph{productivity}, i.e.,  that such  proof rewriting strategies  construct 
strictly increasing  approximations of the cut-free  proof, which can  thus be obtained as a (well-defined) \emph{limit}.  

A major obstacle of this approach arises when the bottommost cut $\rrule$ is below another one $\rrule'$.
In this case, no cut-elimination step can be applied to $\rrule$, so proof rewriting runs into an apparent  stumbling block.
To circumvent  this problem, in~\cite{BaeldeDS16,FortierS13} a special cut-elimination step is introduced, which merges $\rrule$ and $\rrule'$ in a single, generalized cut rule called \emph{multicut}.

In this section we study a continuous cut-elimination method that does not rely on multicut rules, following an alternative idea in which the notion of approximation plays an even more central rule, inspired by the topological approaches to infinite trees~\cite{COURCELLE198395}. To this end, we assume the reader familiar with basic definitions on domain-theory (see, e.g.,~\cite{domains}).

\subsection{Approximating coderivations}\label{subsec:cut-elim-approx}

In this subsection we introduce \emph{open coderivations}, which approximate coderivations. 
Open coderivations form Scott-domains, on top of which we will formally define continuous cut elimination.
Furthermore, we exploit open coderivations to present a  decomposition result for finitely expandable and  progressing coderivations.

\begin{definition}
	We define the set of rules $\opll \dfn \nwpll \cup \set{\zero}$, where $\zero\dfn \vlinf{\zero}{}{\Gamma}{}$  for any sequent $\Gamma$.\footnotemark
	\footnotetext{Previously introduced notions and definitions on coderivations extend to open coderivations in the obvious way,  e.g. the global conditions~\Cref{defn:octhread,def:wR}  and the cut-elimination relation $\cutelim$.}
	We will also refer to $\opll$  as the set of coderivations over $\opll$, which we call \defn{open coderivations}. 
	An open coderivation is \defn{normal} if no cut-elimination step can be applied to it, that is, if one premise of each $\cutr$ is a $\zero$.
	An \defn{open derivation} is a derivation in $\opll$. 
	We denote by $\opll(\Gamma)$ the set of open coderivations with conclusion $\Gamma$, and by $\fapx[\der]$ the set of finite approximations of $\der$.
\end{definition}

\begin{definition}
	Let  
	$\der$ be an open coderivation,
	$\Nodes\subseteq \set{1,2}^*$ be a set of mutually incomparable (w.r.t. the prefix  order) nodes of $\der$,
	and
	$\set{\der'_{\nu}}_{\nu \in \Nodes}$ be a set of open coderivations where 
	$\der'_{\nu}$ has the same conclusion  as the subderivation $\der_\nu$ of $\der$.
	We denote by 
	$\der\set{\der'_\nu/\nu}_{\nu\in \Nodes} =
	\der(\der'_{\nu_1}/\nu_1, \ldots, \der'_{\nu_n}/\nu_n)$, 
	the open coderivation obtained by replacing each $\der_\nu$ with 
	$\der'_\nu$.

	The \defn{pruning} of $\der$ over $\Nodes$
	is the open coderivation $\prun{\der}{\Nodes}=\der\set{\zero/\nu}_{\nu\in \Nodes}$.
	If $\der$ and $\der'$ are two open coderivations, then 
	we say that $\der$ is an \defn{approximation} of $\der'$
	(noted $\der \preceq \der'$) 
	iff 
	$\der= \prun{\der'}{\Nodes}$ for some $\Nodes\subseteq\set{1,2}^*$.
	An approximation is \defn{finite} if it is an open derivation.
\end{definition}

Note that $\der$ and $\prun{\der}{\Nodes}$ (and hence $\der'$ if $\der \preceq \der'$) have the same conclusion.

\begin{proposition}
	\label{prop:scott-domain}
	For any sequent $\Gamma$, the poset $\left(\opll(\Gamma), \preceq \right)$ is a  Scott-domain with least element the open derivation $\zero$ and with maximal elements the coderivations (in $\nwpll$) with conclusion $\Gamma$.
	The compact elements are precisely the open derivations in $\opll(\Gamma)$. 
\end{proposition}

{Cut-elimination steps essentially do not increase the size of open derivations, hence:}


\begin{restatable}{lemma}{cutElimApp}\label{thm:cut-elimApp}
	$\cutelim$ over open derivations is strongly normalizing~and~confluent.
\end{restatable}

\Prog and \FE coderivations can be approximated in a canonical~way.}

\begin{proposition}
	\label{prop:canonicity} 
	If $\der \in \ppll$ is \parsimonious, then there is a finite set $\Nodes\subseteq \set{1,2}^*$ of nodes of $\der$ such that $\prun{\der}{\Nodes}$ is a open derivation and each $v\in\Nodes$ is the root of a \nwbox in $\der$.
\end{proposition}
\begin{proof}
	By~\Cref{cor:simple-structure}, there is a set $\Nodes$ of nodes of $\der$ such that: (i)  each node in $\Nodes$ is the root of a \nwbox, and (ii) any infinite branch of $\der$ contains a node in $\Nodes$. Thus, $\prun{\der}{\Nodes}$ must  be finite by weak K\"{o}nig's lemma, and so is $\Nodes$.
\end{proof}

\begin{definition}
	\label{def:decomp}
	Let $\der\in  \ppll$  be \parsimonious. 
	The \defn{decomposition}
	of $\der$ is the
	(unique) set of nodes $\bord \der=\set{\nu_1, \ldots, \nu_k}$ with $k \in \Nset$ such that 
	$\der_{\nu_i}$ is a $\nwbox$ for all $i\in\intset1k$
	and 
	$\base{\der}\coloneqq\prun{\der}{\bord\der}$ is a minimal (w.r.t. $\preceq$) finite approximation.
\end{definition}

\subsection{Domain-theoretic approach to continuous cut-elimination}\label{subsec:domain}

In this subsection we define \emph{maximal and continuous infinitary cut-elimination strategies} ($\mcices$), special rewriting strategies that  stepwise generate   $\omega$-chains approximating  the  cut-free  version of an open coderivation.   In other words, a $\mcices$  computes a  (Scott-)continuous function from open coderivations to cut-free open coderivations.  
Then, we introduce the \emph{height-by-height} $\mcices$, a notable example of $\mcices$ that will be used for our results, and we show that any two $\mcices$s compute the same (Scott-)continuous function.

In what follows,  $\sigma$ denotes a countable sequence of coderivations, and $\sigma(i)$ denotes the $i+1$-th coderivation in $\sigma$. We denote the length of a sequence $\sigma$ by  $\ell (\sigma)\leq \omega$.

\begin{definition}
	\label{defn:mices}
	An \defn{infinitary cut elimination strategy} (or $\ices$ for short) is a family $\sigma= \{\sigma_\der\}_{\der \in \opll}$ where, for all $\der \in \opll$,  $\sigma_\der$ is a sequence of {open} coderivations such that  $\sigma_\der(0)= \der$ and $\sigma_\der(i)\cutelim\sigma_{\der}(i+1)$ for all $0 \leq i < \ell(\sigma_\der)$.
	Given a  $\ices$   $\sigma$,  we define  the function $f_\sigma \colon \opll(\Gamma) \to \opll(\Gamma)$ as 
	$f_\sigma (\der)\dfn \bigsqcup_{i=0}^{\ell(\sigma_\der)}\cf{\sigma_{\der}(i)}$
	where $\cf{\der_i}$ is the greatest cut-free approximation of $\der_i$ (w.r.t.~$\preceq$)\footnotemark
	\footnotetext{$f_\sigma$ is well-defined, as  $(\cf{\sigma_\der(i)})_{0 \leq i < \ell{(\sigma_\der)}}$ is an $\omega$-chain in $\opll$ and so its sup exists by~\Cref{prop:scott-domain}.}.
	An $\ices$  $\sigma$ is a \defn{$\mcices$} if it is: 
	\begin{itemize}
		\item \defn{maximal}: $\sigma_{\der}(\ell(\sigma_\der))$ is normal for any open derivation $\der$  ($\ell(\sigma_\der)<\omega$ by~\Cref{thm:cut-elimApp});

		\item \defn{(Scott)-continuous}: $f_\sigma$ is Scott-continuous.
	\end{itemize}
\end{definition}

Roughly, a maximal $\ices$ is a $\ices$ that applies cut-elimination steps to open derivations (i.e., finite approximations) until a normal (possibly cut-free) open derivation is reached.

The following property states that all $\mcices$s induce the same continuous function, an easy consequence of~\Cref{thm:cut-elimApp} and continuity.

\begin{restatable}{proposition}{confluence}\label{prop:confluence}
	 If $\sigma$ and $\sigma'$ are two $\mcices$s, then $f_\sigma= f_{\sigma'}$.
\end{restatable}

Therefore, we define a specific $\mcices$ we use in our proofs, 
where cut-elimination steps are applied in a deterministic way to the minimal reducible $\cutr$-rules.

\newcommand{\hbh}[1][]{\sigma^\infty_{#1}}
\newcommand{\fhbh}[1]{f_{\hbh}{(#1)}}

\begin{definition}
	\label{def:hbh}
	The \defn{height-by-height} $\ices$ is defined as $\hbh=\set{\hbh[\der]}_{\der \in \opll}$ where $\hbh[\dD](0)=\der$ for each $\dD\in\opll$,
	and
	$\hbh[\dD](i+1)$ is the open coderivation obtained by applying a cut-elimination step to the leftmost reducible $\cutr$-rule with minimal height in $\hbh[\dD](i)$.
\end{definition}

\begin{proposition}
	The $\ices$ $\sigma^\infty$ is a $\mcices$.
\end{proposition}
\begin{proof}
	By definition, $\hbh$ is continuous. 
	It is also maximal since, by \Cref{thm:cut-elimApp}, any open derivation $\der$ normalizes in $n_\dD \in \Nset$ steps;
	so, $\ell(\hbh[\dD])=n_\dD$ and $\hbh[\dD](h_\dD)$ is normal.
\end{proof}

We conclude this section by providing the sketch of proof for the continuous cut-elimination theorem, the main contribution of this paper,
establishing a productivity result and showing that continuous cut-elimination preserves all global conditions.

\begin{restatable}[Continuous Cut-Elimination]{theorem}{CCE}\label{thm:CCE} 
	\hfill
	\begin{enumerate}
		\item\label{MEGA:1} If $\der\in \wppll$, then $f_{\hbh}(\der)\in\nwpll$.
		\item\label{MEGA:2} If $\der\in \wppll$ (resp.~$\der\in \ppll$), then so is $f_{\hbh}(\der)$.
		\item\label{MEGA:3} If $\der\in \wppll$ is \parsimonious, then so is $f_{\hbh}(\der)$.
		\item\label{MEGA:4} If $\der\in \nupll$ (resp.~$\der\in \cpll$), then so is $f_{\sigma^\infty} (\der)$.
	\end{enumerate}
\end{restatable}
\begin{proof}[Sketch of proof]
	\hfill
	\begin{enumerate}
		\item 	
		It suffices to prove that for any $h\geq 0$ there is $n_h\geq 0$ such that 
		$\cf{\hbh[\dD](n_h)}$ has a $\zero$-free  bar $\Nodes_h$ of rules in $\set{\axr, \unit, \cprule}$  of height greater than $h$.
		The existence of a starting  bar for $\der=\hbh[\dD](0)$ is ensured by \wprog condition.
		Then, we show how to define bars of greater height through cut-elimination. 
		The key case is when a  $\cprule$-rule in the bar  is eliminated by a $\cutstep{\cprule}{\wnbrule}$ step, in which case we exploit \Cref{prop:cut-elim-preserves-finexp-reg-weakreg} to find such a new bar.
		The crucial property to establish is that only finitely many refinements of a starting bar are needed to find the  $\Nodes_h$, which follows  from  the fact that, by \wprog condition, there is no branch of $\der$ that contains infinitely many consecutive $\wnbrule$ rules.

		\item 
		We prove the result for $\dD\in\ppll$ since the proof for $\dD\in \wppll$ is similar.
		By the previous point, $\fhbh\dD \in \nwpll$.
		By \Cref{prop:cut-elim-preserves-finexp-reg-weakreg} $\hbh[\dD](i)$ is \prog for all $i<\omega$.
		Therefore if $\fhbh\dD$ contains a non-\prog branch $\mathcal{B}$, it must have been stepwise constructed by pushing upward a $\cutr$-rule in $\dD$.
		We can track the occurrences of this $\cutr$-rule in $\hbh[\der]$ to define a sequence
		$\seq{\rrule_0, \rrule_1, \ldots, \rrule_n, \ldots}_{i\leq \omega}$ of 
		$\cutr$-rules such that  $\rrule_i\in\hbh[\dD](i)$  and 
		either $\rrule_{i}= \rrule_{i+1}$ or $\hbh[\dD](i)\cutelim \hbh[\dD](i+1)$ by applying a cut-elimination step on $\rrule_i$ producing $\rrule_{i+1}$. 
		This sequence of $\cutr$ rules must reduce infinitely many occurrences of a formula $\wn \cneg{A}$ (in a same $\wn$-thread) with infinitely many occurrences of a $\oc A$ (in a same $\oc$-thread). 
		That is, there are infinitely many cut-elimination step $\cutstep{\cprule}{\cprule}$ in the $\hbh[\dD]$ producing an infinite progressing $\oc$-thread in $\mathcal{B}$.
		
		\item 
		Similar to the previous point.

		\item 
		Akin to linear logic, we define the  \emph{depth} of a  coderivation  as  the maximal number of nested $\nwboxes$, and we prove that the depth of  progressing and finitely expandable coderivations is always finite.   
		Moreover, by~\Cref{prop:canonicity}, a \wprog and infinitely expandable coderivation $\der$ can be decomposed to a $\nwbox$-free  finite approximation $\base{\der}$ and a series of $\nwboxes$ $\nwpromotion_1, \ldots, \nwpromotion_k$  with smaller depth. 
		Using this property we define by induction on the depth of $\der$ a maximal and \emph{transfinite} $\ices$ reducing the calls of the $\nwboxes$ orderly, that is, reducing the $i$-th call to a $\cutr$-free coderivation before reducing the $i+1$-th one. 
		This transfinite  $\ices$ has the advantage of making apparent the  preservation of  (weak) regularity under cut-elimination:	leveraging on~\Cref{prop:weak-regular-finite-support},  if we reduce a $\nwbox$ with finite support  (resp. a periodic $\nwbox$) via our transfinite $\ices$, then	we  obtain in the limit a cut-free $\nwbox$ with finite support  (resp. a periodic $\nwbox$). 
		We conclude by showing that this transfinite $\ices$ can be compressed to a ($\omega$-long) $\mcices$ using methods studied in~\cite{Terese, Saurin}.%
		\qedhere
	\end{enumerate}
\end{proof}	

By definition (as the sup of cut-free open coderivations) $f_{\sigma^\infty} (\der)$ is cut-free. Each item of \Cref{thm:CCE} say in particular that $f_{\sigma^\infty} (\der)$ is $\zero$-free, which means that $f_{\sigma^\infty} (\der)$ is obtained by eliminating \emph{all} the cuts in $\der$. 
This may not be the case if $\der$ does not fulfill any of the global conditions in the hypotheses of \Cref{thm:CCE}: $f_{\sigma^\infty} (\der)$ is still cut-free but may contain some ``truncating'' $\zero$ that ``prevented'' eliminating some cut in $\der$, as in the example below.

\begin{example}
	For any finite approximation $\der$ of the  (non-weakly \prog, non-finitely expandable) open coderivation  $\zeroder$, we have $f_{\sigma^\infty} (\der)=  \zero$, so  $f_{\sigma^\infty}(\zeroder)= \zero$ by continuity.
\end{example}

\section{Relational semantics for non-wellfounded proofs}
\label{sec:semantics}

{Here} we define a denotational model for $\opll$ based on \emph{relational semantics}, {which interprets an open coderivation as} the union of the interpretations of its finite approximations, as  in~\cite{farzad-infinitary}.
{We show that relational semantics is sound for $\opll$, but not for its extension with digging}.

Relational semantics interprets exponential by finite multisets, denoted by brackets, e.g., $\mset{x_1, \ldots, x_n}$;
$+$ denotes the \emph{multiset union},
$\fmset{X}$ denotes the set of finite multisets over a set $X$. 
To correctly define the semantics of a coderivation, we need to see sequents as \emph{finite sequence} of formulas (taking their order into account), which means that we have to add an \emph{exchange} rule to $\opll$ to swap the order of two consecutive formulas in a sequent.

\begin{definition}
	\label{defn:coderivations-semantics-short}
	We 
	associate with each formula $A$ a \defn{set} $\sem A$ defined as follows:
	\begin{equation*}
		\adjustbox{max width=\textwidth}{$
			\sem{X} 		\dfn D_X 	
			\quad
			\sem{\lone} 	\dfn \set{*}
			\quad
			\sem{A \ltens B}\dfn \sem{A}\times \sem B
			\quad
			\sem{\oc A}   	\dfn \fmset{\sem A}
			\quad
			\sem{ \cneg A } \dfn \sem{A}
		$}
	\end{equation*}
	where $D_X$ is an arbitrary set.
	For a sequent $\Gamma=A_1, \ldots, A_n$, we set
	$\sem{\Gamma} \dfn \sem{A_1 \lpar\cdots\lpar A_n}$.

	Given {$\der \in \pll \cup \opll$} with conclusion $\Gamma$,
	we set $ \sem{\der} \dfn \bigcup_{n \geq 0} \sem{\der}_n \subseteq \sem{\Gamma}$, 
	where $\sem\der_0=\emptyset$
	and, for all $i\in\Nset\setminus\set{0}$, 
	$\sem\der_i$ is defined inductively according to \Cref{fig:sem}.
\end{definition}

\def\queq{=}
\begin{figure*}[t]
	\begin{adjustbox}{max width= \textwidth}$
		\begin{array}{l}
			\fSem{
				\vlinf{\axr}{}{A, \cneg A}{\vlhy{}}
			}{ n} 
			\queq
			\Setdef{(x, x)}{ x \in \sem{A} }
			\qquad
			\fSem{
				\vlderivation{
					\vliin{\cutr}{}{\Gamma, \Delta}{
						\vldr{\der'}{\Gamma, A}
					}{
						\vldr{\der''}{\Delta, \cneg A}
					}
				}
			}{n}
			\queq 
			\Setdef{(\vec x, \vec y)}{
				\exists z\in \sem{A}\mbox{ s.t.} \!
				\begin{array}{c}
					(\vec x, z)\in \fsem{\der'}{n-1} 	\\
					\mbox{ and } 						\\
					(z, \vec y)\in \fsem{\der''}{n-1} 
				\end{array}
			}
			\\[15pt]
			\fSem{
				\vlderivation{
					\vlin{\lbot}{}{ \Gamma, \lbot}{\vldr{\der'}{\Gamma }}
				}
			}{n}
			\queq
			\Setdef{(\vec x, *)}{\vec x\in \fsem{\der'}{n-1}}
			\qquad
			\fSem{
				\vlderivation{
					\vlin{\lpar}{}{\Gamma, A \lpar  B}{
						\vldr{\der'}{\Gamma, A, B}
					}
				}
			}{n}
			\queq
			\Setdef{(\vec x, (y,z))}{(\vec x, y, z)\in \fsem{\der'}{n-1}}
			\\[15pt]
			\fSem{
				\vlinf{\lone}{}{\lone}{}
			}{n}
			\queq
			\set{*}
			\qquad
			\fSem{
				\vlderivation{
					\vliin{\ltens}{}{\Gamma, \Delta, A \ltens B}{
						\vldr{\der'}{\Gamma, A}
					}{
						\vldr{\der''}{\Delta, B}
					}
				}
			}{n} 
			\queq 
			\Setdef{(\vec x, \vec y, (x, y))}{
				\begin{array}{c}
					(\vec x, x)\in \fsem{\der'}{n-1} 	\\
					\mbox{ and } 						\\
					( \vec y, y)\in \fsem{\der''}{n-1} 
				\end{array}
			}
			\qquad
			\fSem{\vlderivation{\zerorule{\Gamma}}}{n}
			\queq
			\emptyset
			\\[10pt]
			\fSem{
				\vlderivation{
					\vlin{\wnwrule}{}{ \Gamma, \wn A}{\vldr{\der'}{\Gamma }}
				}
			}{n}
			\queq
			\Setdef{ (\vec x, \emptymset)}{ \vec x\in \fsem{\der'}{n-1} } 
			\qquad
			\fSem{
				\vlderivation{
					\vlin{\wnbrule}{}{\Gamma, \wn A}{
						\vldr{\der'}{\Gamma, A, \wn A }
					}
				}
			}{n}
			\queq
			\Setdef{ (\vec x, \mset{y}+\mu)}{(\vec x, y, \mu)\in \fsem{\der'}{n-1} } 
			\\[15pt]
			\fSem{
				\vlderivation{
					\vliin{\cprule}{}{\wn \Gamma, \oc A}{
						\vldr{\der'}{ \Gamma,A }
					}{
						\vldr{\der''}{\wn \Gamma,\oc A }
					}
				}
			}{n}
			\queq 
			\Set{(\vec{\emptymset} , \emptymset)} \ \cup \  
			\Setdef{(\mset{x_1}+\mu_1, \ldots , \mset{x_k}+\mu_k, \mset{x}+\mu)}
			{\begin{array}{c}
					(x_1, \ldots,x_k, x )\in \fsem{\der'}{n-1} 
					\\
					\mbox{and}
					\\
					(\mu_1, \ldots,\mu_k, \mu )\in \fsem{\der''}{n-1}
				\end{array}
			}
			%
		\end{array}
		$\end{adjustbox}
	\caption{
		Inductive definition of the set $\fsem{\der}{n}$, for $n>0$.
	}
	\label{fig:sem}
\end{figure*}

\begin{example}
	\label{exmp:examples-of-interpretations}
	For the coderivations $\zeroder$ and  $\wnder$ in \Cref{fig:nonProg}, 
	$\sem{\zeroder}=\sem{\wnder}=\emptyset$.
	For the derivations $\falseder$ and $\trueder$ in \Cref{fig:der:examples-of-derivations}, $\sem{\falseder} = \set{(\emptymset, (x,x)) \mid x \in D_X}$ and $\sem{\trueder} = \set{(\mset{(x,y)},(x,y)) \mid x,y \in D_X}$.
	For the coderivation $\derstream{\cod{i_0},\ldots,\cod{i_n}}$ in \Cref{ex:nonReg} (with $i_j \in \set{\false,\true}$ for all $j \in \Nset$),
	$\sem{\derstream{\cod{i_0},\ldots,\cod{i_n}}} = \Set{\emptymset} \cup \Set{\mset{x_{i_0}, \dots, x_{i_n}} \in \fmset{\sem{\Nat}} \mid n \in \Nset, \ x_{i_j} \!\in \sem{\cod{i_j}} \ \forall \, 0 \leq j \leq n}$.
\end{example}

{By inspecting the cut-elimination steps and by continuity, we can prove the soundness of relational semantics with respect to cut-elimination (\Cref{thm:soundness-semantics}), thanks to the fact the interpretation of a coderivation is the union the interpretations of its finite approximation.}

\begin{restatable}{lemma}{semanticapproximation}
	\label{lem:semantic-approximation}
	{Let $\der\in \opll$. Then,  $\sem{\der} = \bigcup_{\der'\in \fapx[ \der]}\sem{\der'}$.}
\end{restatable}

\begin{restatable}[Soundness]{theorem}{soundness} 
	\label{thm:soundness-semantics}
\begin{enumerate}
	\item\label{p:soundness-semantics-single} Let $\der\in \opll$. If $\der \cutelim \der'$, then $\sem\der=\sem{\der'}$.
	\item Let $\der\in \opll$. If $\sigma$ is a $\mcices$, then $\sem{\der}= \sem{f_\sigma(\der)}$.
\end{enumerate}
\end{restatable}

By \Cref{thm:soundness-semantics} and since cut-free coderivations have non-empty semantics, we have:

\begin{restatable}{corollary}{nonempty}
	\label{cor:progressiveness-non-empty} Let $\der\in \wppll$. Then $\sem{\der}\neq \emptyset$.
\end{restatable}

We define the set of rules $\nwmell \dfn \nwpll \cup \set{\digr}$ where the rule $\digr$ (\defn{digging}) is defined in \Cref{fig:digging}. 
We also denote by $\nwmell$ the set of coderivations over the rules in $\nwmell$.
Relational semantics is naturally extended to $\nwmell$ as shown in \Cref{fig:digging}.

The proof system $\nwmell$ can be seen as a non-wellfounded version of $\mell$.
We show that, as opposed  to several fragments of $\nwpll$, in any good fragment of $\nwmell$ with digging, cut-elimination cannot reduce to cut-free coderivations preserving the relational semantics.

\begin{figure*}[!t]
	\begin{adjustbox}{max width= 0.95\textwidth}$
			\vlinf{\digr}{}{ {\Gamma}, {\wn A}}{ {\Gamma}, {\wn \wn A}}
			\qquad\qquad
			\fSem{
				\vlderivation{
					\vlin{\digr}{}{\Gamma, \wn A}{
						\vldr{\der'}{\Gamma, \wn \wn A }
					}
				}
			}{0}
			\queq
			\emptyset
			\qquad
			\fSem{
				\vlderivation{
					\vlin{\digr}{}{\Gamma, \wn A}{
						\vldr{\der'}{\Gamma, \wn \wn A }
					}
				}
			}{n}
			\queq
			\Setdef{ \left(\vec x, \displaystyle\sum_{i=1}^m \mu_i\right)}{(\vec x, \mset{\mu_1, \dots, \mu_m}) \in \fsem{\der'}{n-1}, \ m \in \Nset }
		$\end{adjustbox}
	\caption{The rule $\digr$ and its interpretation in the relational semantics ($n > 0$).}
	\label{fig:digging}
\end{figure*}

\begin{theorem}\label{thm:no-semantics-digging}
	Let $\sysX \subseteq \nwmell$ contain non-wellfounded coderivations with $\digr$. 
	Let $\cutelimnhl$ be a cut-elimination relation on $\sysX$ containing $\cutelim$ in \Cref{fig:cut-elim-finitary-lin,fig:cut-elim-finitary-comm,fig:cut-elim-pll}
	and reducing every coderivation in $\sysX$ to a cut-free one.
	Then, $\cutelimnhl$ does not preserve relational~semantics.
\end{theorem}

\begin{proof}
Consider the coderivations $\digder$ and $\widehat{\digder}$ below, where
$\der = \derstream{\falseder,\trueder,\falseder,\trueder}$,
and $\der_i = \derstream{\cod{k_{0}^{i}}, \dots, \cod{k_{n}^i}}$ and $k_{j}^i \in \set{\false,\true}$ for all $i,j \in \Nset$ (see also \Cref{ex:nonReg}).

\vspace*{-1.5\baselineskip}
\begin{equation}\label{eqn:digging}
	\hfill
	\adjustbox{max width=.75\textwidth}{$
		\digder \dfn
		\vlderivation{
			\vliin{\cutr}{}{\oc\oc\Nat}{\vldr{\der}{\oc\Nat}}{
				\vlin{\digr}{}
				{{\wn \cneg{\Nat}, \oc\oc\Nat}}
				{\vlin{\axr}{}{\wn \wn \cneg{\Nat}, \oc\oc\Nat}{}}
			}
		}
		\qquad\qquad
		\widehat{\digder} 
		\dfn
		\vlderivation{
			\vliin{\cprule}{}{\oc\oc\Nat}{\vldr{\der_0}{\oc\Nat}}{
				\vliin{\cprule}{}{{\oc\oc\Nat}}{
					\vldr{\der_1}{\oc \Nat}
				}{				 			
					\vliin{\cprule}{}{\reflectbox{$\ddots$}}{\vldr{\der_n}{\oc\Nat}}{\vlin{\cprule}{}{{\oc\oc\Nat}}{\vlhy{\vdots}}}
				}
			}
		}
	$}
	\hfill
\end{equation}	
Coderivations $\widehat{\digder}$ are the only cut-free ones with conclusion $\oc \oc \Nat$. 
Therefore, for whatever definition of the cut-elimination steps concerning $\digr$, necessarily $\digder$ will reduce to $\widehat{\digder}$.

Let $\hat{\false}$ be the only element of $\sem{\cod{\false}}$, and $\hat{\true}$ be any element of $\sem{\cod{\true}}$ (see \Cref{exmp:examples-of-interpretations}).
Note that $\hat{\false} \neq \hat{\true}$. 
It is easy to verify that $\mset{\mset{\hat{\false}},\mset{\hat{\false}}}, \mset{\mset{\hat{\true}},\mset{\hat{\true}}} \notin \sem{\digder}$,
 since $\mset{\hat{\false},\hat{\false}}, \mset{\hat{\true},\hat{\true}} \notin \sem{\der}$ (see \Cref{exmp:examples-of-interpretations}).
 Concerning $\sem{\widehat{\digder}}$, we notice that, since $k^0_0, k^1_0, k^2_0\in \set{\false,\true}$,  either  $k^0_0=k^1_0$ or  $k^1_0=k^2_0$ or $k^2_0=k^0_0$. In the first case, we have  $\mset{\mset{k^0_0}, \mset{k^1_0}}\in \sem{\widehat{\digder}}$, in the second case we have $\mset{\mset{k^1_0}, \mset{k^2_0}}\in \sem{\widehat{\digder}}$, and in the last case we have $\mset{\mset{k^2_0}, \mset{k^0_0}}\in \sem{\widehat{\digder}}$.
\end{proof}

\section{Conclusion and future work}

For future research, we envisage extending {our contributions} in many directions. First,  our notion of finite approximation seems intimately related with  that of Taylor expansion from \emph{differential linear logic} ($\mathsf{DiLL}$)~\cite{ehr:introDiLL}, where   the rule  $\zero$  (quite like the rule $0$ from $\mathsf{DiLL}$)  serves  to model approximations of \emph{boxes}. This connection with Taylor expansions  becomes even more apparent in Mazza's original systems for parsimonious logic~\cite{Mazza15,MazzaT15}, which comprise co-absorption  and co-weakening rules typical of $\mathsf{DiLL}$. These considerations deserve  further investigations. Secondly, building on a series of recent works in 
  \emph{Cyclic Implicit Complexity}, i.e., implicit computational complexity in the setting of circular and non-wellfounded proof theory~\cite{Curzi023,CurziDas}, we are currently working on second-order extensions of  $\nupll$ and $\cpll$   to characterize the complexity classes $\ppoly$  and $\ptime$ (see~\cite{BLANK}). These results would reformulate in a non-wellfounded  setting the characterization  of $\ppoly$ presented in~\cite{MazzaT15}.




\bibliography{biblo}

\begin{thebibliography}{10}

\bibitem{domains}
Roberto~M. Amadio and Pierre{-}Louis Curien.
\newblock {\em Domains and lambda-calculi}, volume~46 of {\em Cambridge tracts in theoretical computer science}.
\newblock Cambridge University Press, 1998.

\bibitem{arora_barak_2009}
Sanjeev Arora and Boaz Barak.
\newblock {\em Computational Complexity: A Modern Approach}.
\newblock Cambridge University Press, 2009.
\newblock \href {https://doi.org/10.1017/CBO9780511804090} {\path{doi:10.1017/CBO9780511804090}}.

\bibitem{BaeldeDS16}
David Baelde, Amina Doumane, and Alexis Saurin.
\newblock Infinitary proof theory: the multiplicative additive case.
\newblock In Jean{-}Marc Talbot and Laurent Regnier, editors, {\em 25th {EACSL} Annual Conference on Computer Science Logic, {CSL} 2016, August 29 - September 1, 2016, Marseille, France}, volume~62 of {\em LIPIcs}, pages 42:1--42:17. Schloss Dagstuhl - Leibniz-Zentrum f{\"{u}}r Informatik, 2016.
\newblock \href {https://doi.org/10.4230/LIPIcs.CSL.2016.42} {\path{doi:10.4230/LIPIcs.CSL.2016.42}}.

\bibitem{BaeldeM07}
David Baelde and Dale Miller.
\newblock Least and greatest fixed points in linear logic.
\newblock In Nachum Dershowitz and Andrei Voronkov, editors, {\em Logic for Programming, Artificial Intelligence, and Reasoning, 14th International Conference, {LPAR} 2007, Yerevan, Armenia, October 15-19, 2007, Proceedings}, volume 4790 of {\em Lecture Notes in Computer Science}, pages 92--106. Springer, 2007.
\newblock \href {https://doi.org/10.1007/978-3-540-75560-9\_9} {\path{doi:10.1007/978-3-540-75560-9\_9}}.

\bibitem{brotherston2011sequent}
James Brotherston and Alex Simpson.
\newblock Sequent calculi for induction and infinite descent.
\newblock {\em Journal of Logic and Computation}, 21(6):1177--1216, 2011.

\bibitem{COURCELLE198395}
Bruno Courcelle.
\newblock Fundamental properties of infinite trees.
\newblock {\em Theoretical Computer Science}, 25(2):95--169, 1983.
\newblock URL: \url{https://www.sciencedirect.com/science/article/pii/0304397583900592}, \href {https://doi.org/10.1016/0304-3975(83)90059-2} {\path{doi:10.1016/0304-3975(83)90059-2}}.

\bibitem{CurziDas}
Gianluca Curzi and Anupam Das.
\newblock Cyclic implicit complexity.
\newblock In {\em Proceedings of the 37th Annual ACM/IEEE Symposium on Logic in Computer Science}, LICS '22, New York, NY, USA, 2022. Association for Computing Machinery.
\newblock \href {https://doi.org/10.1145/3531130.3533340} {\path{doi:10.1145/3531130.3533340}}.

\bibitem{Curzi023}
Gianluca Curzi and Anupam Das.
\newblock Non-uniform complexity via non-wellfounded proofs.
\newblock In Bartek Klin and Elaine Pimentel, editors, {\em 31st {EACSL} Annual Conference on Computer Science Logic, {CSL} 2023, February 13-16, 2023, Warsaw, Poland}, volume 252 of {\em LIPIcs}, pages 16:1--16:18. Schloss Dagstuhl - Leibniz-Zentrum f{\"{u}}r Informatik, 2023.
\newblock \href {https://doi.org/10.4230/LIPIcs.CSL.2023.16} {\path{doi:10.4230/LIPIcs.CSL.2023.16}}.

\bibitem{DanosJ03}
Vincent Danos and Jean{-}Baptiste Joinet.
\newblock Linear logic and elementary time.
\newblock {\em Inf. Comput.}, 183(1):123--137, 2003.
\newblock \href {https://doi.org/10.1016/S0890-5401(03)00010-5} {\path{doi:10.1016/S0890-5401(03)00010-5}}.

\bibitem{Das2021}
Anupam Das.
\newblock On the logical strength of confluence and normalisation for cyclic proofs.
\newblock In Naoki Kobayashi, editor, {\em 6th International Conference on Formal Structures for Computation and Deduction, {FSCD} 2021, July 17-24, 2021, Buenos Aires, Argentina (Virtual Conference)}, volume 195 of {\em LIPIcs}, pages 29:1--29:23. Schloss Dagstuhl - Leibniz-Zentrum f{\"{u}}r Informatik, 2021.
\newblock \href {https://doi.org/10.4230/LIPIcs.FSCD.2021.29} {\path{doi:10.4230/LIPIcs.FSCD.2021.29}}.

\bibitem{dax2006proof}
Christian Dax, Martin Hofmann, and Martin Lange.
\newblock A proof system for the linear time $\mu$-calculus.
\newblock In {\em International Conference on Foundations of Software Technology and Theoretical Computer Science}, pages 273--284. Springer, 2006.

\bibitem{ehr:introDiLL}
Thomas Ehrhard.
\newblock An introduction to differential linear logic: proof-nets, models and antiderivatives, 2016.
\newblock \href {https://arxiv.org/abs/1606.01642} {\path{arXiv:1606.01642}}.

\bibitem{Farzad}
Thomas Ehrhard and Farzad Jafar{-}Rahmani.
\newblock On the denotational semantics of linear logic with least and greatest fixed points of formulas.
\newblock {\em CoRR}, abs/1906.05593, 2019.
\newblock URL: \url{http://arxiv.org/abs/1906.05593}, \href {https://arxiv.org/abs/1906.05593} {\path{arXiv:1906.05593}}.

\bibitem{farzad-infinitary}
Thomas Ehrhard, Farzad Jafarrahmani, and Alexis Saurin.
\newblock {On relation between totality semantic and syntactic validity}.
\newblock In {\em {5th International Workshop on Trends in Linear Logic and Applications (TLLA 2021)}}, Rome (virtual), Italy, June 2021.
\newblock URL: \url{https://hal-lirmm.ccsd.cnrs.fr/lirmm-03271408}.

\bibitem{fortier2013cuts}
J{\'e}r{\^o}me Fortier and Luigi Santocanale.
\newblock Cuts for circular proofs: semantics and cut-elimination.
\newblock In {\em Computer Science Logic 2013 (CSL 2013)}. Schloss Dagstuhl-Leibniz-Zentrum fuer Informatik, 2013.

\bibitem{FortierS13}
J{\'{e}}r{\^{o}}me Fortier and Luigi Santocanale.
\newblock Cuts for circular proofs: semantics and cut-elimination.
\newblock In Simona Ronchi~Della Rocca, editor, {\em Computer Science Logic 2013 {(CSL} 2013), {CSL} 2013, September 2-5, 2013, Torino, Italy}, volume~23 of {\em LIPIcs}, pages 248--262. Schloss Dagstuhl - Leibniz-Zentrum f{\"{u}}r Informatik, 2013.
\newblock \href {https://doi.org/10.4230/LIPIcs.CSL.2013.248} {\path{doi:10.4230/LIPIcs.CSL.2013.248}}.

\bibitem{gir:ll}
Jean-Yves Girard.
\newblock Linear logic.
\newblock {\em Theoretical Computer Science}, 50(1):1--101, 1987.
\newblock \href {https://doi.org/10.1016/0304-3975(87)90045-4} {\path{doi:10.1016/0304-3975(87)90045-4}}.

\bibitem{girard:98}
Jean-Yves Girard.
\newblock Light linear logic.
\newblock {\em Information and Computation}, 143(2):175--204, 1998.
\newblock URL: \url{https://www.sciencedirect.com/science/article/pii/S0890540198927006}, \href {https://doi.org/10.1006/inco.1998.2700} {\path{doi:10.1006/inco.1998.2700}}.

\bibitem{Kuperberg-Pous21}
Denis Kuperberg, Laureline Pinault, and Damien Pous.
\newblock Cyclic proofs, system {T}, and the power of contraction.
\newblock {\em Proc. {ACM} Program. Lang.}, 5({POPL}):1--28, 2021.
\newblock \href {https://doi.org/10.1145/3434282} {\path{doi:10.1145/3434282}}.

\bibitem{lafont:soft}
Yves Lafont.
\newblock Soft linear logic and polynomial time.
\newblock {\em Theoretical Computer Science}, 318(1):163--180, 2004.
\newblock Implicit Computational Complexity.
\newblock \href {https://doi.org/10.1016/j.tcs.2003.10.018} {\path{doi:10.1016/j.tcs.2003.10.018}}.

\bibitem{BLANK}
Gianluca~Curzi Matteo~Acclavio and Giulio Guerrieri.
\newblock Non-uniform polynomial time via non-wellfounded parsimonious proofs.
\newblock URL: \url{http://gianlucacurzi.com/Non-uniform-polynomial-time-via-non-wellfounded-parsimonious-proofs.pdf}.

\bibitem{Mazza15}
Damiano Mazza.
\newblock Simple parsimonious types and logarithmic space.
\newblock In Stephan Kreutzer, editor, {\em 24th {EACSL} Annual Conference on Computer Science Logic, {CSL} 2015, September 7-10, 2015, Berlin, Germany}, volume~41 of {\em LIPIcs}, pages 24--40. Schloss Dagstuhl - Leibniz-Zentrum f{\"{u}}r Informatik, 2015.
\newblock \href {https://doi.org/10.4230/LIPIcs.CSL.2015.24} {\path{doi:10.4230/LIPIcs.CSL.2015.24}}.

\bibitem{MazzaT15}
Damiano Mazza and Kazushige Terui.
\newblock Parsimonious types and non-uniform computation.
\newblock In Magn{\'{u}}s~M. Halld{\'{o}}rsson, Kazuo Iwama, Naoki Kobayashi, and Bettina Speckmann, editors, {\em Automata, Languages, and Programming - 42nd International Colloquium, {ICALP} 2015, Kyoto, Japan, July 6-10, 2015, Proceedings, Part {II}}, volume 9135 of {\em Lecture Notes in Computer Science}, pages 350--361. Springer, 2015.
\newblock \href {https://doi.org/10.1007/978-3-662-47666-6\_28} {\path{doi:10.1007/978-3-662-47666-6\_28}}.

\bibitem{NigamM09}
Vivek Nigam and Dale Miller.
\newblock Algorithmic specifications in linear logic with subexponentials.
\newblock In Ant{\'{o}}nio Porto and Francisco~Javier L{\'{o}}pez{-}Fraguas, editors, {\em Proceedings of the 11th International {ACM} {SIGPLAN} Conference on Principles and Practice of Declarative Programming, September 7-9, 2009, Coimbra, Portugal}, pages 129--140. {ACM}, 2009.
\newblock \href {https://doi.org/10.1145/1599410.1599427} {\path{doi:10.1145/1599410.1599427}}.

\bibitem{niwinski1996games}
Damian Niwi{\'n}ski and Igor Walukiewicz.
\newblock Games for the $\mu$-calculus.
\newblock {\em Theoretical Computer Science}, 163(1-2):99--116, 1996.

\bibitem{pag:tor:StrongNorm}
Michele Pagani and Lorenzo Tortora~de Falco.
\newblock Strong normalization property for second order linear logic.
\newblock {\em Theor. Comput. Sci.}, 411(2):410–444, jan 2010.
\newblock \href {https://doi.org/10.1016/j.tcs.2009.07.053} {\path{doi:10.1016/j.tcs.2009.07.053}}.

\bibitem{quatrini:phd}
Myriam Quatrini.
\newblock {\em S{\'e}mantique coh{\'e}rente des exponentielles: de la logique lin{\'e}aire {\`a} la logique classique}.
\newblock PhD thesis, Aix-Marseille 2, 1995.

\bibitem{Roversi}
Luca Roversi and Luca Vercelli.
\newblock Safe recursion on notation into a light logic by levels.
\newblock In Patrick Baillot, editor, {\em Proceedings International Workshop on Developments in Implicit Computational complExity, {DICE} 2010, Paphos, Cyprus, 27-28th March 2010}, volume~23 of {\em {EPTCS}}, pages 63--77, 2010.
\newblock \href {https://doi.org/10.4204/EPTCS.23.5} {\path{doi:10.4204/EPTCS.23.5}}.

\bibitem{Saurin}
Alexis Saurin.
\newblock {A linear perspective on cut-elimination for non-wellfounded sequent calculi with least and greatest fixed points (extended version)}.
\newblock working paper or preprint, 2023.
\newblock URL: \url{https://hal.science/hal-04169137}.

\bibitem{Simpson17}
Alex Simpson.
\newblock Cyclic arithmetic is equivalent to peano arithmetic.
\newblock In Javier Esparza and Andrzej~S. Murawski, editors, {\em Foundations of Software Science and Computation Structures - 20th International Conference, {FOSSACS} 2017, Held as Part of the European Joint Conferences on Theory and Practice of Software, {ETAPS} 2017, Uppsala, Sweden, April 22-29, 2017, Proceedings}, volume 10203 of {\em Lecture Notes in Computer Science}, pages 283--300, 2017.
\newblock \href {https://doi.org/10.1007/978-3-662-54458-7\_17} {\path{doi:10.1007/978-3-662-54458-7\_17}}.

\bibitem{Curry-Howard}
Morten~Heine S\o{}rensen and Pawel Urzyczyn.
\newblock {\em Lectures on the Curry-Howard Isomorphism, Volume 149 (Studies in Logic and the Foundations of Mathematics)}.
\newblock Elsevier Science Inc., USA, 2006.

\bibitem{Terese}
Terese.
\newblock {\em Term rewriting systems}, volume~55 of {\em Cambridge tracts in theoretical computer science}.
\newblock Cambridge University Press, 2003.

\bibitem{troelstra_schwichtenberg_2000}
A.~S. Troelstra and H.~Schwichtenberg.
\newblock {\em Basic Proof Theory}.
\newblock Cambridge Tracts in Theoretical Computer Science. Cambridge University Press, 2 edition, 2000.
\newblock \href {https://doi.org/10.1017/CBO9781139168717} {\path{doi:10.1017/CBO9781139168717}}.

\end{thebibliography}

\clearpage

\appendix

\section{Appendix of \Cref{sec:parsimoniousLogic}}\label{app:3}

\begin{figure*}[t]
	\adjustbox{max width=\textwidth}{$\begin{array}{c|c}
		\ptom{
			\vlderivation{
				\vlin{\rrule}{}
				{\Gamma}
				{\vldr{\der}{\Gamma'}}
			}
		}
		\dfn
		\vlderivation{
			\vlin{\rrule}{}
			{\Gamma}
			{\vldr{\ptomp{\der}}{\Gamma'}}
		}
		\quad
		\ptom{
			\vlderivation{
				\vliin{\trule}{}{\Gamma}{
					\vldr{\der_1}{\Gamma_1}
				}{
					\vldr{\der_2}{\Gamma_2}
				}
			}
		}
		\dfn
		\vlderivation{
			\vliin{\trule}{}{\Gamma}{
				\vldr{\ptomp{\der_1}}{\Gamma_1}
			}{
				\vldr{\ptomp{\der_2}}{\Gamma_2}
			}
		}
		&
		\ptom{
			\vlderivation{\vlin{\wnbrule}{}{\Gamma, \wn A}{
				{\vldr{\der}{\Gamma, A,\wn A}}
			}}
		}
		\dfn
		\vlderivation{
			\vlin{\wncrule}{}{\Gamma, \wn A}{
				\vlin{\wnderr}{}{\Gamma, \wn A,\wn A}{					{\vldr{\ptomp\der}{\Gamma, A,\wn A}}
				}
			}
		}
		\quad
		\ptom{
			\vlderivation{\vlin{\fprule}{}{\wn \Gamma, \oc A}{
				{\vldr{\der}{\Gamma, A}}
			}}
		}
		\dfn
		\vlderivation{
			\vlin{\prule}{}{\wn \Gamma, \oc A}{\vliq{\wnderr}{}{\wn \Gamma, A}{
					{\vldr{\ptomp{\der}}{\Gamma,A}}
				}
			}
		}
		\\
		\mbox{for all } 
		\rrule\in\set{\axr,\lpar,\lone,\lbot,\wnwrule}
		\quand
		\trule\in\set{\cutr,\ltens}
	\end{array}$}
	\caption{
		Translation $\ptom{\cdot}$ from $\pll$ to $\mell$.
	}
	\label{fig:translations-mell}
\end{figure*}

\begin{figure}[t]
	\centering
	\adjustbox{max width=\textwidth}{$
		\begin{tikzcd}
			\vlderivation{
				\vliin{\cutr}{}{
					\wn\Gamma, \wn\Delta, \oc B
				}{
					\vlin{\fprule}{}{\wn \Gamma, \oc A}{\vlhy{\Gamma, A}}
				}{
					\vlin{\fprule}{}{\wn \cneg A, \wn\Delta, \oc B}{
						\vlhy{\cneg A, \Delta, B}
					}
				}
			}
			\arrow[r, "\spadesuit"]
			\arrow[dd, "\cutstep{\wnbrule}{\fprule}"]
			&
			\vlderivation{
				\vliin{\cutr}{}{
					\wn\Gamma, \wn\Delta, \oc B
				}{
					\vlin{\prule}{}{\wn \Gamma, \oc A}{
						\vliq{\wnderr}{}{\wn\Gamma, A}{\vlhy{\Gamma, A}}
					}
				}{
					\vlin{\prule}{}{\wn \cneg A, \wn\Delta, \oc B}{
						\vliq{\wnderr}{}{\wn\cneg A, \wn\Delta, B}{\vlhy{\cneg A, \Delta, B}} 
					}
				}
			}
			\arrow[d,"\cutstep{\fprule}{\fprule}"]
			\\&
			\vlderivation{
				\vlin{\prule}{}{
					\wn\Gamma, \wn\Delta, \oc B
				}{
					\vliin{\cutr}{}{
						\wn\Gamma, \wn\Delta, B
					}{
						\vlin{\prule}{}{\wn \Gamma, \oc A}{
							\vliq{\wnderr}{}{\wn\Gamma, A}{\vlhy{\Gamma, A}}
						}
					}{
						\vliq{\wnderr}{}{\cneg A, \wn\Delta, B}{
							\vlin{\wnderr}{}{\wn\cneg A, \Delta, B}{
								\vlhy{\cneg A, \Delta, B}
							} 
						}
					}
				}
			}
			\arrow[d, two heads, "\cutstep{\fprule}{\wnderr}"]
			\\
			\vlderivation{
				\vlin{\fprule}{}{
					\wn\Gamma, \wn\Delta, \oc B
				}{
					\vliin{\cutr}{}{
						\Gamma, \Delta, B
					}{
						\vlhy{\Gamma, A}
					}{
						\vlhy{\cneg A, \Delta, B}
					}
				}
			}
			\arrow[r, "\spadesuit"]
			&
			\vlderivation{
				\vlin{\prule}{}{
					\wn\Gamma, \wn\Delta, \oc B
				}{
					\vliq{\wnderr}{}{\wn\Gamma, \wn\Delta, B}{
						\vliin{\cutr}{}{
							\Gamma, \Delta, B
						}{
							\vlhy{\Gamma, A}
						}{
							\vlhy{\cneg A, \Delta, B}
						}
					}
				}
			}
		\end{tikzcd}
	$}
	\caption{Commutation of the $\cutstep{\wnbrule}{\fprule}$ step and $\ptom{\cdot}$.}
	\label{fig:pllVSmell1}
\end{figure}
\begin{figure}[t]
	\centering
	\adjustbox{max width=\textwidth}{$		
			\begin{tikzcd}
				\vlderivation{
					\vliin{\cutr}{}{\wn \Gamma, \Delta}{
						\vlin{\fprule}{}{\wn \Gamma, \oc A}{
							\vlhy{\Gamma, A}
						}
					}{
						\vlin{\wnbrule}{}{\Delta, \wn \cneg A}{
							\vlhy{\Delta,\cneg {A}\!, \wn\cneg A}
						}
					}
				}
				\arrow[r, "\spadesuit"]
				\arrow[ddd, "\cutstep{\fprule}{\wnbrule}"]
				&
				\vlderivation{
					\vliin{\cutr}{}{\wn \Gamma, \Delta}{
						\vlin{\prule}{}{\wn \Gamma, \oc A}{
							\vliq{\wnderr}{}{\wn\Gamma, A}{\vlhy{\Gamma, A}}
						}
					}{
						\vlin{\wncrule}{}{\Delta, \wn \cneg A}{
							\vlin{\wnderr}{}{
								\Delta,\wn\cneg A, \wn\cneg A
							}{
								\vlhy{\Delta,\cneg {A}, \wn\cneg A}
							}
						}
					}
				}
				\arrow[d,"\cutstep{\fprule}{\wncrule}"]
				\\&
				\vlderivation{
					\vliq{\wncrule}{}{\wn \Gamma, \Delta}{
						\vliin{\cutr}{}{\wn \Gamma,\wn \Gamma, \Delta}{
							\vlin{\prule}{}{\wn \Gamma, \oc A}{
								\vliq{\wnderr}{}{\wn\Gamma, A}{\vlhy{\Gamma, A}}
							}
						}{
							\vliin{\cutr}{}{\wn \Gamma,\Delta, \wn\cneg A}{
								\vlin{\prule}{}{\wn \Gamma, \oc A}{
									\vliq{\wnderr}{}{\wn\Gamma, A}{\vlhy{\Gamma, A}}
								}
							}{
								\vlin{\wnderr}{}{\Delta,\wn\cneg A, \wn\cneg A}{
									\vlhy{\Delta,\cneg {A}, \wn\cneg A}
								}
							}
						}
					}
				}
				\arrow[d, "\mbox{commutative step}"]
				\\
				&
				\vlderivation{
					\vliq{\wncrule}{}{\wn \Gamma, \Delta}{
						\vliin{\cutr}{}{\wn \Gamma,\wn \Gamma, \Delta}{
							\vlin{\prule}{}{\wn \Gamma, \oc A}{
								\vliq{\wnderr}{}{\wn\Gamma, A}{\vlhy{\Gamma, A}}
							}
						}{
							\vlin{\wnderr}{}{\wn \Gamma,\Delta, \wn\cneg A}{
								\vliin{\cutr}{}{\wn \Gamma,\Delta, \cneg A}{
									\vlin{\prule}{}{\wn \Gamma, \oc A}{
										\vliq{\wnderr}{}{\wn\Gamma, A}{\vlhy{\Gamma, A}} 
									}
								}{
									\vlhy{\Delta,\cneg {A}, \wn\cneg A}
								}
							}
						}
					}
				}
				\arrow[d, two heads, "\cutstep{\fprule}{\wnderr}"]
				\\
				\vlderivation{
					\vliq{\wnbrule}{}{\wn\Gamma, \Delta}{
						\vliin{\cutr}{}{\Gamma, \wn \Gamma,\Delta }{
							\vlhy{\Gamma, A}
						}{
							\vliin{\cutr}{}{\wn\Gamma,  \Delta, \cneg A}{
								\vlin{\fprule}{}{\wn \Gamma,\oc A}{\vlhy{\Gamma , A}}
							}{
								\vlhy{\Delta, \cneg {A}\!, \wn \cneg A}
							}
						}
					}
				}
				\arrow[r,"\spadesuit"]
				&
				\vlderivation{
					\vliq{\wncrule}{}{\wn\Gamma, \Delta}{
						\vliq{\wnderr}{}{\wn\Gamma,\wn \Gamma, \Delta}{
							\vliin{\cutr}{}{\Gamma, \wn \Gamma,\Delta }{
								\vlhy{\Gamma, A}
							}{
								\vliin{\cutr}{}{\wn\Gamma,  \Delta, \cneg A}{
									\vlin{\prule}{}{\wn \Gamma, \oc A}{
										\vliq{\wnderr}{}{\wn\Gamma, A}{\vlhy{\Gamma, A}} 
									}
								}{
									\vlhy{\Delta, \cneg {A}\!, \wn \cneg A}
								}
							}
						}
					}
				}
			\end{tikzcd}
		$}
	\caption{
		Commutation of the $\cutstep{\fprule}{\wnbrule}$ step with $\ptom\cdot$.
	}
	\label{fig:pllVSmell2}
\end{figure}

\cutelimPLL*
\begin{proof}
	We recall the sequent calculus for (propositional) \emph{multiplicative exponential linear logic}  
	$\mell=\set{\axr,\ltens,\lpar,\lone, \lbot,\cutr,\prule,\wnwrule,\wnderr,\wncrule}$
	where
	the \defn{promotion} ($\prule$), \defn{dereliction} ($\wnderr$), \defn{contraction} ($\wncrule$) rules are defined as follows:
	\begin{equation}
		\hfill
		\vlinf{\prule}{}{\wn \Gamma, \oc A}{\wn \Gamma, A}
		\qquad
		\vlinf{\wnderr}{}{\Gamma, \wn A}{ \Gamma, A}
		\qquad
		\vlinf{\wncrule}{}{ \Gamma, \wn  A}{ \Gamma, \wn A, \wn A}
		\hfill
	\end{equation}
	
	We also denote by $\mell$ the set of derivations over the rules in $\mell$, 
	and we map each derivation in $\dD\in\pll$ to 
	a derivation in $\ptom\dD\in\mell$
	$\ptom{\cdot	} \colon \pll \to \mell$
	defined in \Cref{fig:translations-mell} by induction on derivations.
	
	In order to prove that the following diagram commute, 
	$$
	\begin{tikzcd}
		\der
		\arrow[r, "\spadesuit"]
		\arrow[d]
		&
		\ptomp\der 
		\arrow[d, "\mbox{possibly many steps}"]
		\\
		\der'
		\arrow[r, "\spadesuit"]
		&
		\ptom{\der'}
	\end{tikzcd}
	$$	
	Each cut-elimination step in $\pll$ corresponds to a cut-elimination step in $\mell$ except the ones in \Cref{fig:pllVSmell1,fig:pllVSmell2}, 
	where a cut-elimination step in $\pll$ can be simulated by a sequence of cut-elimination steps in $\mell$.
	In these Figures each macro-step denoted by $\mathrel{\mathrlap{\rightarrow}\mkern.1mu\rightarrow}$
	involves a unique step from \Cref{fig:cut-elim-finitary-exp,fig:cut-elim-finitary-comm} (the one marked) and certain additional commutative cut-elimination steps of the following form below
\begin{equation}\label{eq:additional-commutative}
	\footnotesize
	\vlderivation{
		\vliin{\cutr}{}{\Gamma, \Delta, \wn B}{
			\vlhy{\Gamma, A\!\!\!\!\!}
		}{
			\vlin{\wnderr}{}{\cneg A, \Delta,\wn B}{
				\vlhy{\cneg A, \Delta, B}
			}
		}
	}
	\cutelim\!\!
	\vlderivation{
		\vlin{\wnderr}{}{\Gamma, \Delta, \wn B}{
			\vliin{\cutr}{}{\Gamma, \Delta, B}{
				\vlhy{\Gamma, A}
			}{
				\vlhy{\cneg A,\Delta, B}
			}
		}
	}
	\quad
	\quad
	\vlderivation{
		\vlin{\wncrule}{}{\Gamma, \wn A, \wn B}{
			\vlin{\wnderr}{}{\Gamma, \wn A, \wn B, \wn B}{
				\vlhy{\Gamma, A, \wn B, \wn B}
			}
		}
	}
	\cutelim
	\vlderivation{
		\vlin{\wnderr}{}{\Gamma, \wn A, \wn B}{
			\vlin{\wncrule}{}{\Gamma,  A, \wn B}{
				\vlhy{\Gamma, A, \wn B, \wn B}
			}
		}
	}
	\end{equation}
	which {push $\wnderr$ down a cut and create an alternating chain of $\wnderr$ and $\wncrule$ (such additional steps are natural to consider since they involve rule permutations of independent rules and  would appear whenever a $\cutr$-rule would interact with the $\wn$-formula introduced by the $\wnderr$-rule).
	Thus, the derivation in $\mell$ obtained by (standard and additional) cut-elimination from $\ptomp{\der}$ is exactly the translation $\ptom{\der'}$ of the derivation $\der'$ in $\pll$ obtained after a cut-elimination step from $\der$.
	According to the definition of $\ptom{\cdot}$, if $\ptom{\der'}$ is cut-free~then~so~is~$\der'$.}
	
	The termination of cut-elimination in $\mell$ with this additional commutative step follows from the result in $\mell$~\cite{pag:tor:StrongNorm}.
	{Indeed, to the usual measure $m$ that decreases after each standard cut-elimination step in $\mell$ (and remains unchanged after each additional step in \eqref{eq:additional-commutative}), we can add the sum $d$ of the heights of the $\wnderr$ rules in a derivation, which decreases after each step in \eqref{eq:additional-commutative}. Thus, the measure $(m,d)$ with the lexicographical order decreases after each (standard or additional) cut-elimination step in $\mell$.
}
\end{proof}

\begin{figure*}[!thb]
	\adjustbox{max width=\textwidth}{$\begin{array}{c}
			\vlderivation{
				\vliin{\cutr}{}{\wn \Gamma, \wn \Delta, \oc B}{
					\vlin{\nuprule}{}{\wn \Gamma, \oc A}{
						\vlhy{\left\{ \vlderivation{ \vldr{\der_i}{\Gamma, A}}\right\}_{i \in \Nset} }
					}
				}{
					\vlin{\nuprule}{}{\wn \cneg{A}\!, \wn \Delta, \oc B}{
						\vlhy{\left\{ \vlderivation{ \vldr{\der'_i}{\cneg{A}\!, \Delta, B}}\right\}_{i \in \Nset} }
					}
				}
			}
			\cutelim\!\!
			\vlderivation{
				\vlin{\nuprule}{}{\wn \Gamma, \wn \Delta, \oc B}{
					\vlhy{
						\left\{  \vlderivation{ 
							\vliin{\cutr}{}{\Gamma, \Delta, B}
							{
								\vldr{\der_i}{\Gamma, A}
							}{
								\vldr{\der'_i}{\cneg{A}\!, \Delta, B}
							}
						} 	\right\}_{i \in \Nset}
					}
				}
			}
			\qquad
			\vlderivation{
				\vliin{\cutr}{}{\wn \Gamma, \Delta}{
					\vlin{\nuprule}{}{\wn \Gamma, \oc A}{
						\vlhy{\left\{ \vlderivation{ \vldr{\der_i}{\Gamma, A}}\right\}_{i \in \Nset} }
					}
				}{
					\vlin{\wnwrule}{}{\Delta,\wn \cneg A}{\vlhy{\Delta}}
				}
			}
			\cutelim
			\vlderivation{
				\vliq{\wnwrule}{}{\wn\Gamma,\Delta}{\vlhy{\Delta}}
			}
			%
			\\\\
			\vlderivation{
				\vliin{\cutr}{}{\wn \Gamma, \Delta}{
					\vlin{\nuprule}{}{\wn \Gamma, \oc A}{
						\vlhy{\left\{ \vlderivation{ \vldr{\der_i}{\Gamma, A}}\right\}_{i \in \Nset} }
					}
				}{
					\vlin{\wnbrule}{}{\Delta, \wn \cneg A}{
						\vlhy{\Delta,\cneg{A}\!, \wn\cneg A}
					}
				}
			}
			\cutelim\!
			\vlderivation{
				\vliq{\wnbrule}{}{\wn\Gamma, \Delta}{
					\vliin{\cutr}{}{\Gamma, \wn \Gamma,\Delta }{
						\vldr{\der_0}{\Gamma, A}
					}{
						\vliin{\cutr}{}{\wn\Gamma,  \Delta, \cneg A}{
							\vlin{\nuprule}{}{\wn \Gamma,\oc A}
							{
								\vlhy{\left\{ \vlderivation{ \vldr{\der_{i+1}}{\Gamma, A}}\right\}_{i \in \Nset} }
							}
						}{
							\vlhy{\Delta, \cneg{A}\!, \wn \cneg A}
						}
					}
				}
			}
		\end{array}
		$}
	\caption{Exponential cut-elimination steps in $\dpll$.}
	\label{fig:damiano-cut-elimination}
\end{figure*}

\section{Proofs of \Cref{sec:coder}}

Akin to linear logic, the \emph{depth} of a coderivation {is} the maximal number of nested $\nwbox$s.

\begin{definition}\label{def:depth}
	Let {$\der \in \nwpll$}.
	The \defn{nesting level of a sequent {occurrence}} $\Gamma$ in $\der$ 
	is the number $\nestl \der \Gamma$ 
	of nodes below it that are the root of a call of a $\nwbox$. 
	The \defn{nesting level of a formula (occurrence)} $A$ in $\der$, noted $\nestl \der A$, is the nesting level of the sequent that contain that formula. 
	The \defn{nesting level of a rule} $\rrule$ in $\der$, noted $\nestl \der \rrule$
	(resp.~\defn{of a sub-coderivation} $\der'$ of $\der${, noted $\nestl \der{\der'}$}),
	is the nesting level of the conclusion of $\rrule$
	(resp.~conclusion of $\der'$).

	The \defn{depth of $\der$}
	is 
	$\depth \der\coloneqq \sup_{\rrule \in \der} \set{\nestl \der \rrule} \in \Nset \cup \{\infty\}$.
\end{definition}

\begin{remark}\label{rem:nesting}
	All calls of a $\nwbox$ have the same nesting level.
	Moreover, each of the sequents of its main branch have nesting level $0$.
\end{remark}	

Cut-elimination $\cutelim$ on $\nwpll$ enjoys the following property.
\begin{lemma}\label{lem:depth}
	Let $\der, \der' \in \nwpll$. 
	If $\der \cutelim \der'$ then $\depth{\der} \geq \depth{\der'}$.
\end{lemma}

\begin{proof}
	By inspection of the cut-elimination steps in \Cref{fig:cut-elim-finitary-lin,fig:cut-elim-finitary-comm,fig:cut-elim-pll}.
\end{proof}

\begin{lemma}
	\label{lem:finite-depth}
	If $\der \in \ppll$ then $\depth{\der} \in \Nset$. 
\end{lemma}

\begin{proof}
	If $\der$ had infinite depth, there would exist an infinite branch that goes left at $\cprule$  infinitely  often. This branch cannot contain a (progressing) $\oc$-thread.
\end{proof}

\translations*
\begin{proof}\hfill
	\begin{enumerate}
		\item By straightforward induction on $\der \in \pll$ (resp. $\der \in \dpll$).
		\item By \Cref{lem:finite-depth}, $\depth{\der} \in \Nset$. We can then prove the statement by induction on $\depth{\der}$.
		\qedhere
	\end{enumerate}
\end{proof}

\preserves*

\begin{proof}
	By inspection of the cut-elimination steps defined in \Cref{fig:cut-elim-finitary-comm,fig:cut-elim-finitary-lin,fig:cut-elim-pll}.
\end{proof}

\section{Proofs of \Cref{sec:ccutelim}}\label{app:cut-elim}

\cutElimApp*
	\begin{proof}
		For $\der$ an open derivation, let $\weightcp{\der}$ be the number of $\cprule$ in $\der$ and $\heightcut{\der}$ be the sum of the sizes of all subderivations
		of $\der$ whose root is the conclusion of a $\cutr$ rule.
		If~$\der \cutelim \der'$~via:
		\begin{itemize}
			\item a commutative cut-elimination step, then $\weightcp{\der} = \weightcp{\der'}$, $\size{\der} = \size{\der'}$ and $\heightcut{\der} > \heightcut{\der'}$;
			\item a multiplicative cut-elimination (\Cref{fig:cut-elim-finitary-lin}),  then $\weightcp{\der} = \weightcp{\der'}$ and $\size{\der} > \size{\der'}$;
			\item an exponential cut-elimination step (\Cref{fig:cut-elim-pll}),  then $\weightcp{\der} > \weightcp{\der'}$.
		\end{itemize}
		Since the lexicographic order over the triples $\seq{\weightcp{\der},\size{\der},\heightcut{\der}} \in \omega^3$ is wellfounded,  we conclude that there is no infinite sequence $\seq{\der_i}_{i \in \Nset}$ such that $\der_0 = \der$ and $\der_{i} \cutelim \der_{i+1}$.
		
		Finally, since cut-elimination $\cutelim$ is strongly normalizing over open derivations 
		and it is locally confluent by  inspection of critical pairs, by Newman's lemma it is also confluent. 
	\end{proof}
	
	\confluence*
	\begin{proof}
		For any open derivation $\der$, since  $\sigma$ and $\sigma'$ are maximal, we have that $\sigma_\der(\ell(\sigma_\der))$ and $\sigma'_\der(\ell(\sigma'_\der))$ are normal, and so $\sigma_\der(\ell(\sigma_\der))=\sigma'_\der(\ell(\sigma'_\der))$ by~\Cref{thm:cut-elimApp}. Hence:
		$$
		f_\sigma(\der)= \cf{\sigma_\der(\ell(\sigma_\der))}=\cf{\sigma'_\der(\ell(\sigma'_\der))}=f_{\sigma'}(\der)
		$$
		Now, let $\der$ be an open coderivation, and let $F(\der)$ be the set of its finite approximations. Since by~\Cref{prop:scott-domain}  $\opll$ is a  Scott-domain, it is  also algebraic, so that we  have $\der= \bigsqcup_{\der' \in F(\der)}\der'$. By continuity of $f_\sigma$ and $f_{\sigma'}$ we have:
		$f_\sigma(\der)=\bigsqcup_{\der' \in F(\der)}f_\sigma(\der')=\bigsqcup_{\der' \in F(\der)}f_{\sigma'}(\der')=f_{\sigma'}(\der)$.
	\end{proof}

	\begin{definition}[$\cprule$-chains]\label{defn:cprule-chain} Let $\sigma=\{\sigma_\der\}_{\der\in \opll}$ be a $\ices$ and let $\der \in \opll$. For any $i$, we write $\rrule_i \rightsquigarrow \rrule_{i+1}$ if  $\rrule_i$ is a $\cprule$ rule in $\sigma_\der(i)$, 
		$\rrule_{i+1}$ is a $\cprule$ rule in $\sigma_\der(i+1)$, and $\sigma_\der(i) \cutelim \sigma_\der(i+1)$ is applied to a $\cutr$ rule  immediately below $\rrule_i$ and produces  $\rrule_{i+1}$. 	A  \defn{$\cprule$-chain} in $\sigma_\der$   is any sequence of $\cprule$ rules $(\rrule_i)_{i< \alpha}$ with $\alpha \leq \ell(\sigma_\der) $ such that:
		\begin{itemize}
			\item for all $i\geq 0$, $\rrule_i$ is in $\sigma_\der(i)$
			\item either $\rrule_i=\rrule_{i+1}$ or $\rrule_i \rightsquigarrow \rrule_{i+1}$. 
		\end{itemize}
	\end{definition}
	
	\begin{remark}\label{rem:uniqueness-of-cprule-chain}
	Let $\sigma=\{\sigma_\der\}_{\der\in \opll}$ be a $\ices$. 
		If $\rrule$ is a $\cprule$ rule in $\der$, then there is a \emph{unique} maximal $\cprule$-chain \mbox{$(\rrule_i)_{i<\alpha}$ in  $\sigma_\der$ with ($\alpha \leq \ell(\sigma_\der)$ and)   $\rrule=\rrule_0$}.
	\end{remark}


	The following lemma establishes a productivity result for the height-by-height $\mcices$.



\begin{restatable}{lemma}{WPROG}\label{thm:preservation-weakly-progressive} 
	If $\der\in \wppll$, then $f_{\sigma^\infty} (\der)\in \nwpll$.
\end{restatable}	
\begin{proof}
		Let $\der$ be a weakly \prog coderivation. Since $\der$ is by assumption $\zero$-free and no cut-elimination rule introduces $\zero$, we can assume  $\ell{(\sigma^\infty_\der)}= \omega$.	In what follows, we shorten $\sigma^\infty_\der(i)$ with $\der_i$, so $\der_0=\der$. We show a stronger statement: for any $h \geq 0$ there is a $n_h \geq 0$ such that $\cf{\der_{n_h}}$ has a $\zero$-free  bar $\Nodes_h$ of height greater than $h$. By definition, this will allow us to conclude  that $f_{\sigma^\infty}(\der)= \bigsqcup_i \cf{\der_i}$ is $\zero$-free.
		
		Let $h \geq 0$. 	We define a procedure computing $\Nodes_h$ divided into rounds, where at the $j$-th round we compute $\Nodes^j_h$.	 At round $0$  we set  $\Nodes^0_h$ to be  a bar across $\der$ with height greater than $h$ containing only rules in $\{\axr, \unit, \cprule\}$  (such a bar exists by weakly \prog).  At the $j$-th  round with $j>0$, the procedure constructs $\Nodes_h^{j}$ from $\Nodes^{j-1}_h$. It  analyses the first node of $\Nodes^{j-1}_h$ that has not been considered in previous rounds (giving priority to nodes with highest prefix order)\footnote{More precisely,  $\nu= a_0\ldots a_n  < c_0\ldots c_m=\nu'$ with $a_i, c_i \in \{1,2\}$ iff there is $i_0 \leq m$ such that $a_i \leq c_i$ for any $0 \leq i \leq i_0$ and $a _{i_0+1} <c_{i_0+1}$.}. Let $\rrule^j$ be such a node. We only consider the case where $\rrule^j$ is a  $\cprule$ rule.  We consider the  $\cprule$-chain $(\rrule^j_i)_i$ such that $\rrule^j= \rrule^j_0$ (which is unique by~\Cref{rem:uniqueness-of-cprule-chain}). If there is a  least $i_0$ such that $\rrule^j_{i_{0}}\rightsquigarrow \rrule^j_{i_{0}+1}$ (so that $ \rrule^j_{i_{0}+1}$ is  produced by applying a  principal cut-elimination step to $ \rrule^j_{i_{0}}$,  and $\rrule^j_{i}=\rrule^j_{i+1}$ for all $i <i_0$), then  we have three cases:
		\begin{itemize}
			\item If  the cut-elimination step has shape $\cutstep{\cprule}{ \cprule}$ then we set $\Nodes_h^{j}\dfn (\Nodes^{j-1}_h \setminus \{\rrule^j_{}\}) \cup \{ \rrule^j_{i_{0}+1}\}$ and we move to the next  round. 
			\item 	If the  cut-elimination step has shape $\cutstep \cprule  \wnwrule$ then we set $\Nodes^j_h \dfn \Nodes^{j-1}_{h}\setminus \{\rrule^j\}$ and we  move to the next round. 
			\item Otherwise,  the  cut-elimination step has shape $\cutstep \cprule \wnbrule$. Let $\der'$ be the coderivation of $\sigma_\der$ containing the rule $\rrule^j_{i_0}$,  let  
			$\nu$ be the node of $\der'$ that is conclusion of $\rrule^j_{i_0}$, and let $\mathcal{U}^j_h$  be a suitable bar of $\der'_\nu$ at height $>0$ containing only rules in $\{\axr, \unit, \cprule\}$. This bar exists by weakly \prog of $\der$ and the fact that  weak \prog is preserved under finite cut-elimination by~\Cref{prop:cut-elim-preserves-finexp-reg-weakreg}. We  set $\Nodes^j_h \dfn (\Nodes^{j-1}_{h}\setminus \{\rrule^j\})\cup \mathcal{U}^j_h$ and we move to the next round. 
		\end{itemize}
		If no such such $\rrule^j_{i_{0}}\rightsquigarrow \rrule^j_{i_{0}+1}$ exists (so $\rrule^j_i=\rrule^j_{i+1}$ for all $i$)    we move to the next round. 
		
		By construction, if the procedure terminates, it computes the set of nodes $\Nodes_h$ such that, for some $k\geq 0$ sufficiently large,  $\Nodes_h$  defines  a bar across any $\der_i$  in the sequence $\sigma_\der$ for all $i \geq k$. This means that there exists $n_h\geq k$ such that $\cf{\der_{n_h}}$ contains that bar. So we have to show that the procedure terminates. Since bars are finite, this boils down to proving that there are only finitely many rounds. Suppose towards contradiction that this is not the case. This can only happen when there are infinitely many distinct $\cprule$ rules $(\rrule_i)_i$ in a branch $\mathcal{B}_j$ of $\der$  and infinitely many distinct $\wnbrule$ rules $(\rrule'_i)_i$ in a branch $\mathcal{B}'_j$ of $\der$ such that in $\sigma_\der^\infty$:
		\begin{enumerate}
			\item \label{enum:contradition1} each $\rrule_i$ is eventually cut with $\rrule'_i$,
			\item \label{enum:contradition2} each $\rrule_i$ is never cut with a $\cprule$ rule.
		\end{enumerate}  
		Notice that the assumption that the rules in $(\rrule'_i)_i$ belong to the same  branch $\mathcal{B}'_j$ causes no loss of generality, since the  height-by-height $\mcices$ reduces the cut $\cutstep{\rrule_i}{\rrule'_i}$ before any other cut above these rules. By~\Cref{enum:contradition1} $\mathcal{B}'_j$ is infinite, and by~\Cref{enum:contradition2}  it is eventually $\cprule$-free, contradicting weakly \prog of $\der$. 
\end{proof}

The following notion is the analogue of (\emph{multi})\emph{cut reduction sequences} from~\cite{BaeldeDS16}.

\begin{definition}[Cut-chains]\label{defn:cut-chain} Let $\sigma=\{\sigma_\der\}_{\der \in \opll}$ be a $\ices$ and let $\der \in \opll$. For any $i$, we write $\rrule_i \mapsto \rrule_{i+1}$ if  $\rrule_i$ is a $\cutr$ rule in $\sigma_\der(i)$, 
	$\rrule_{i+1}$ is a $\cutr$ rule in $\sigma_\der(i+1)$, and $\sigma_\der(i)\cutelim \sigma_\der(i+1)$ is applied to $\rrule_i$ producing  $\rrule_{i+1}$. 	A  \defn{$\cutr$-chain} in $\sigma_\der$   is any sequence of $\cutr$ rules $(\rrule_i)_{i< \alpha }$ with $\alpha \leq \ell(\sigma_\der)$ such that:
	\begin{itemize}
		\item for all $i\geq 0$, $\rrule_i$ is in $\sigma_\der(i)$
		\item either $\rrule_i=\rrule_{i+1}$ or $\rrule_i \mapsto \rrule_{i+1}$. 
	\end{itemize}
\end{definition}
	
\begin{remark}\label{rem:infinite-cut-chain-eventually-oc-formulas}
 Let $\sigma=\{\sigma_\der\}_{\der \in \opll}$ be a $\ices$ , and let $(\rrule_i)_{i}$ be an infinite cut chain in $\sigma_\der$ such that $(A_i, \cneg{A}_i)$ is the pair of cut formulas of $\rrule_i$. There is $i_0 \geq 0$ such that, for all $i \geq i_0$, $A_i$ is a $\oc$-formula (and $\cneg{A}_i$ is a  $\wn$-formula).
\end{remark}

\begin{remark}\label{rem:at-most-one-infinite-thread}
	Any branch  $\mathcal{B}$ in a \prog coderivation $\der$  contains at most (and hence exactly) one \prog $\oc$-thread. As a consequence,  any infinite $\oc$-thread $\tau$ of a branch $\mathcal{B}$ in a \prog coderivation $\der$ must be \prog.  Indeed, let $\tau$ and $\tau'$ be two infinite $\oc$-threads, and let us show that  $\tau= \tau'$. Since $\mathcal{B}$ is \prog, it contains infinitely many $\cprule$ rules $(\rrule_i)_i$, so that there exists $n \geq 0$ such that both $\tau$ and $\tau'$ contain formulas below $\rrule_i$. Since the conclusion of $\rrule_i$ has exactly one $\oc$-formula and $\tau$ is infinite, both $\tau$ and $\tau'$ must contains that formula, so that $\tau=\tau'$ by maximality of $\oc$-threads.
\end{remark}
	
\begin{restatable}{lemma}{PROG}\label{thm:preservation-progressiveness} 
	\hfill
	\begin{enumerate}
		\item \label{enum:point1}	 If $\der\in \wppll$ (resp. $\der \in \ppll$), then so is $f_{\sigma^\infty}(\der)$.
		\item \label{enum:point2}  If $\der\in \wppll$ is finitely expandable, then so is $f_{\sigma^\infty} (\der)$.
	\end{enumerate}
\end{restatable}
\begin{proof}
	Let us prove~\Cref{enum:point1}.	Let $\der$ be a  \prog open coderivation, and let  us shorten $\sigma^\infty_\der(i)$ with $\der_i$, so $\der_0=\der$.  By~\Cref{prop:cut-elim-preserves-finexp-reg-weakreg} we can assume that $\ell(\sigma^\infty_\der)=\omega$.

	We want to show that for any infinite $\cutr$-chain  $(\rrule_i)_{i < \omega}$ in  $\sigma^\infty_\der$ such that:
	\begin{enumerate}
		\item[(I)] \label{enum:cut-chain-conditions1} $\rrule_0$ is a $\cutr$ rule with minimal height  in $\der$
		\item[(II)]\label{enum:cut-chain-conditions2} $\pi(\rrule_i) =  a_0a_1\ldots a_{n_i}$ is the address of $\rrule_i$ in $\der_i$ (with $n_i \leq n_{i+1}$),  
	\end{enumerate}
	there exists $0 \leq k_0\leq  n_0$ and an infinite family $\tau^*\dfn (C_i)_{k_0 \leq i }$ of occurrences of a $\oc$-formula satisfying the following properties:
	\begin{enumerate}[a]
		\item  \label{enum:cut-chain-thread-construction1} $\tau^*_i\dfn (C_j)_{k_0 \leq j \leq n_i}$ is a \octhread in $\pi(\rrule_i)$
		\item \label{enum:cut-chain-thread-construction2} for any $m \geq 0$ there is $i$ such that $\tau^*_{i}$ has $m$ progressing points.
	\end{enumerate}
	
	Notice that the property above allows us to conclude. Indeed,  let $\mathcal{B}$ be an infinite branch of $f_{\sigma^\infty}(\der)$. If $\mathcal{B}$ is in some $\der_i$, then it is \prog by~\Cref{prop:cut-elim-preserves-finexp-reg-weakreg}. Otherwise, there exists an infinite $\cutr$-chain   $(\rrule_i)_{i < \omega}$ in $\sigma^\infty_\der$ satisfying~\Cref{enum:cut-chain-thread-construction1},~\Cref{enum:cut-chain-thread-construction2} and $\mathcal{B}=\{ \pi(\rrule_i) \ \vert \ i \geq 0\}$. By~\Cref{enum:cut-chain-thread-construction1} and~\Cref{enum:cut-chain-thread-construction2} there is an infinite family $(C_i)_{n_0 \leq i }$ of occurrences of a $\oc$-formula that defines a \prog \octhread of $\mathcal{B}$. 
	
	So, let  $(\rrule_i)_{i < \omega}$ be a $\cutr$-chain with minimal height such that:
	\begin{itemize}
		\item  the premises of $\rrule_i$ are conclusions of the rules  $\rrule'_i$ and   $\rrule''_i$
		\item $(A_i, \cneg{A}_i)$  are the cut formulas of $\rrule_i$
		\item $\pi(\rrule_i) =  a_0a_1\ldots a_{n_i}$ is the address of $\rrule_i$ in $\der_i$ 
	\end{itemize}
	By~\Cref{rem:infinite-cut-chain-eventually-oc-formulas}, we can assume w.l.o.g.~that $A_{i}=\oc B$ and  $\cneg{A}_i=\wn \cneg{B}$. It is easy to see that $\tau\dfn (A_i)_i$ is an infinite $\oc$-thread of some branch $\mathcal{B}'$ of $\der$ and that $\tau'\dfn (\cneg{A}_i)$ is an infinite $\wn$-thread of some branch $\mathcal{B}''$ in $\der$. Moreover, by~\Cref{rem:at-most-one-infinite-thread} and by \progness of $\der$, $\tau$ is \prog. 
	This means that there are infinitely  many $i$ such that  $\rrule_i'=\rrule''_i=\cprule$ (so that $A_i$ is the principal $\oc$-formula of $\rrule'_i$ and $\cneg{A}_{i+1}$ is an auxiliary $\wn$-formula of $\rrule''_i$)  and $\rrule_i \mapsto \rrule_{i+1}$. Let $\tau''\dfn (C_i)_i$ be the \prog $\oc$-thread of $\mathcal{B}''$. 
	Since $\rrule_0$ is a $\cutr$ with minimal height, and the minimal height  $\cutr$ rules never decreases during cut-elimination,  all cuts $\rrule_i$ in the cut chain have minimal height. 
	This means that the first formula of $\tau''$, i.e.,  $C_0$,  is not a cut formula, and so it is in the end-sequent of $\der$. 
	It is easy to see that the cut-elimination rules never affect $\tau''$ (and its \prog points) while pushing upward the $\cutr$ rules.  This means that  we can construct  $\tau''$ satisfying the  properties ~\Cref{enum:cut-chain-thread-construction1} and~\Cref{enum:cut-chain-thread-construction2}.

	Let us now prove~\Cref{enum:point2}.		Since $f_{\sigma^\infty}(\der)$ is cut-free we only have to show that all of its infinite branches have only finitely many  $\wnbrule$ rules each.
	Let $\der$ be a  \parsimonious open coderivation, and   let  us shorten $\sigma^\infty_\der(i)$ with $\der_i$, so $\der_0=\der$.  By~\Cref{prop:cut-elim-preserves-finexp-reg-weakreg} we can assume that $\ell(\sigma^\infty_\der)=\omega$.

	We want to show that for any infinite $\cutr$-chain  $(\rrule_i)_{i < \omega}$ in  $\sigma^\infty_\der$ such that:
	\begin{enumerate}
		\item[(I)] \label{enum:cut-chain-fconditions1} $\rrule_0$ is a $\cutr$ rule with minimal height  in $\der$
		\item[(II)]\label{enum:cut-chain-fconditions2} $\pi(\rrule_i) =  a_0a_1\ldots a_{n_i}$ is the address of $\rrule_i$ in $\der_i$ (with $n_i \leq n_{i+1}$),  
	\end{enumerate}
	the branch  $\mathcal{B}=\{ \pi(\rrule_i) \ \vert \ i \geq 0\}$ of $f(\der)$ has only finitely many distinct $\wnbrule$ rules. 
	Note that the property above allows us to conclude. Indeed,  let $\mathcal{B}$ be an infinite branch of $f_{\sigma^\infty}(\der)$. If $\mathcal{B}$ is in some $\der_i$, then it is finitely expandable by~\Cref{prop:cut-elim-preserves-finexp-reg-weakreg}. Otherwise, there is an infinite $\cutr$-chain   $(\rrule_i)_{i < \omega}$ in $\sigma^\infty_\der$ such that $\mathcal{B}=\{ \pi(\rrule_i) \ \vert \ i \geq 0\}$, so  we are done by the above property.

	Thus, let  $(\rrule_i)_{i < \omega}$ be a $\cutr$-chain with minimal height such that:
	\begin{itemize}
		\item  the premises of $\rrule_i$ are conclusions of the rules  $\rrule'_i$ and   $\rrule''_i$;
		\item $(A_i, \cneg{A}_i)$  are the cut formulas of $\rrule_i$;
		\item $\pi(\rrule_i) =  a_0a_1\ldots a_{n_i}$ is the address of $\rrule_i$ in $\der_i$.
	\end{itemize}
	By~\Cref{rem:infinite-cut-chain-eventually-oc-formulas}, we can assume w.l.o.g.~that $A_{i}=\oc B$ and  $\cneg{A}_i=\wn \cneg{B}$. It is easy to see that $\tau\dfn (A_i)_i$ is an infinite $\oc$-thread of some branch $\mathcal{B}'$ of $\der$ and that $\tau'\dfn (\cneg{A}_i)$ is an infinite $\wn$-thread of some branch $\mathcal{B}''$ in $\der$. 
	
	Let us suppose towards contradiction that $\mathcal{B}$ has infinitely many $\wnbrule$ rules. This means that, for any $k$ there is $n_k$ such that $\pi(\rrule_{n_k})$ contains $k$ $\wnbrule$ rules. Since $\der$ is finitely expandable,  there must be  infinitely many $i \geq 0$ such that $\rrule_i \mapsto \rrule_{i+1}$ is obtained by applying the cut-elimination step $\cutstep{\cprule}{\wnbrule}$. But this would mean that the $\wn$-thread $\tau'$ contains infinitely many principal rules for $\wnbrule$ rule, and so $\mathcal{B}''$ would contain infinitely many $\wnbrule$ rules, contradicting finite expandability of $\der$.
\end{proof}

\begin{proposition}\label{prop:base-of-limit-finitely-computable}
	Let $\der\in \nupll$ (resp.~$\cpll$). Then $f_{\sigma^\infty}(\der)$  admits a decomposition, and $\base{f_{\sigma^\infty}(\der)}= \base {\sigma^\infty_\der(n)}$  for some $n \geq 0$.
\end{proposition}
\begin{proof}
	By~\Cref{thm:preservation-weakly-progressive} and~\Cref{thm:preservation-progressiveness}, 
	$f_{\sigma^\infty}(\der)$ is a cut free ($\zero$-free) coderivation and finitely expandable coderivation. By~\Cref{prop:canonicity} $f_{\sigma^\infty}(\der)$ admits a decomposition $\bord \der= \{v_1, \ldots, v_k\}$. 	By continuity, this means that there  is $n\geq 0$  such that  $\base {\sigma^\infty_\der(n)}=\base{f_{\sigma^\infty}(\der)}$. Note that $\base {\sigma^\infty_\der(n)}$ exists by ~\Cref{prop:cut-elim-preserves-finexp-reg-weakreg,prop:canonicity}.
\end{proof}

\begin{restatable}{lemma}{WR}\label{thm:preservation-weak-regularity} 
	If $\der\in \nupll$ (resp.~$\der\in \cpll$), then so is $f_{\sigma^\infty} (\der)$.
\end{restatable}
\begin{proof}
We define a  maximal and transfinite $\ices$ $\sigma=\{\sigma_\der\}_{\der\in\opll} $  preserving weak regularity, and show that this strategy can be ``compressed'' to a  $\mcices$, $\sigma^*=\{\sigma^*_\der\}_{\der\in\opll} $ , along the lines of~\cite{Saurin}. We then conclude since $f_{\sigma^\infty}=f_{\sigma^*}$ by~\Cref{prop:confluence}.  
So let $\der \in \nupll$.  By induction on $d=\depth{\der}$ (which is finite by~\Cref{lem:finite-depth}) we define  $\sigma_\der=(\der_i)_i$     such that:
	\begin{enumerate}[(a)]
		\item \label{enum:conditions-transfiniteA} For any limit ordinal $\lambda\leq \ell(\sigma_\der)$:
		\begin{enumerate}[(i)]
			\item  \label{enum:conditions-trasfinite1}  $\bigsqcup_{i<\lambda}\der'_i=\der_\lambda$ for some $\der'_i $ finite approximations  of $\der_i$. 
			\item  \label{enum:conditions-trasfinite2} If $h_i$ is the height of the $\cutr$  reduced at the $i$-th step of $\sigma_\der$ then $\lim{i < \lambda }{h_i}=\infty$. 
		\end{enumerate}
		\item\label{enum:conditions-transfinite-cut-free} $\der_{\ell(\sigma_\der)}$ is cut free
		\item \label{enum:conditions-transfinite3}  $\der_{\ell(\sigma_\der)}$ is weakly regular. 
	\end{enumerate}
	\begin{itemize}
		\item If $d=0$ then by~\Cref{prop:canonicity} $\der$ is an open derivation, so that by~\Cref{thm:cut-elimApp} there is a maximal cut-elimination sequence that rewrites $\der$ to a normal open coderivation. In particular, the latter is also cut-free because $\der$ is $\zero$-free and so every cut can be eventually eliminated. We set $\sigma_\der$ to be such a cut-elimination sequence. By construction, $\sigma_\der$  satisfies~\Cref{enum:conditions-trasfinite1}-\ref{enum:conditions-trasfinite2} and~\Cref{enum:conditions-transfinite-cut-free}. Moreover, by~\Cref{prop:cut-elim-preserves-finexp-reg-weakreg} $\sigma_\der$ satisfies~\Cref{enum:conditions-transfinite3}
		\item  If $d>0$ then by~\Cref{prop:base-of-limit-finitely-computable} there is $n \geq 0$ such that $\base{f_{\sigma^\infty}(\der)}= \base {\sigma^\infty_\der(n)}$
		By construction $\sigma^\infty_\der(n)$ has the following structure:
		$$
		\toks0={0.3}
		\vlderivation{
			\vltrf{\sigma^\infty_\der(n)}{\Gamma}
			{\vlhy{
					\left\{
					\cutr(\mathfrak{G}'_i, \mathfrak{G}''_i)	\dfn \vlderivation{
						\vliin{\cutr}{}{\wn \Delta_i, \wn \Sigma_i, \oc A_i}
						{
							\vltr{\mathfrak{G}'_i}{\wn \Delta_i, \oc B_i}{\vlhy{\ \ }}{\vlhy{\ \  }}{\vlhy{\ \ }}
						}{
							\vltr{\mathfrak{G}''_{i}}{\wn \cneg{B}_i, \wn \Sigma_i, \oc A_i}{\vlhy{\ \  }}{\vlhy{\ \ }}{\vlhy{\ \ }}
						}
					}
					\right\}_{1 \leq i \leq l}
				}
			}
			{
				\vlhy{}
			}
			{\vlhy{
					\left\{
					\vlderivation{
						\vltr{\mathfrak{G}'''_{i}}{\wn \Theta_i, \oc C_i}{\vlhy{\ \  }}{\vlhy{\ \  }}{\vlhy{\ \ }}
					}
					\right\}_{1 \leq i \leq m}
				}
			}
			{\the\toks0}
		}
		$$
		for some $\nwbox$s $\mathfrak{G}'_i, \mathfrak{G}_i'', \mathfrak{G}_i'''$.  For any $1 \leq i \leq l$, let  $\sigma_i$ be the  $\mcices$ that applies only cut-elimination steps for $\cutstep{\cprule}{\cprule}$ and that rewrites $\cutr(\mathfrak{G}'_i, \mathfrak{G}''_i)$ to the following coderivation:
		$$
			\vlderivation{
				\vliin{\cprule}{}{\wn \Delta_i, \wn \Sigma_i, \oc A_i}{\vltr{\cutr(\mathfrak{G}'_i(1), \mathfrak{G}''_i(1))}{\Delta_i, \Sigma_i, A_i}{\vlhy{\ \ \ \ }}{\vlhy{\ \ \ \ }}{\vlhy{\ \ \ \ }}}{
					\vliin{\cprule}{}{{\wn \Delta_i, \wn \Sigma_i, \oc A_i}}{
						\vltr{\cutr(\mathfrak{G}'_i(2), \mathfrak{G}''_i(2))}{\Delta_i, \Sigma_i, A_i}{\vlhy{\ \ \ \ }}{\vlhy{\ \ \ \ }}{\vlhy{\ \ \ \ }}
					}{				 			
						\vliin{\cprule}{}{\reflectbox{$\ddots$}}{\vltr{\cutr(\mathfrak{G}'_i(n), \mathfrak{G}''_i(n))}{\Delta_i, \Sigma_i, A_i}{\vlhy{\ \ \ \ }}{\vlhy{\ \ \ \  }}{\vlhy{\ \ \ \ }}}{
							\vlin{\cprule}{}{{\wn \Delta_i, \wn \Sigma_i, \oc A_i}}{\vlhy{\vdots}}
						}
					}
				}
			}
		$$
		where:
		$$
			\cutr(\mathfrak{G}'_i(j), \mathfrak{G}''_i(j)) \dfn \vlderivation{
				\vliin{\cutr}{}{\Delta_i, \Sigma_i, A_i}
				{
					\vltr{\mathfrak{G}'_i(j)}{\Delta_i, B_i}{\vlhy{\ \ }}{\vlhy{\ \ }}{\vlhy{\ \ }}
				}
				{
					\vltr{\mathfrak{G}''_i(j)}{\cneg{B}_i, \Sigma_i, A_i}{\vlhy{\ \  }}{\vlhy{\ \ }}{\vlhy{\ \  }}
				}
			}
		$$
		By induction hypothesis, for any $j \geq 0$ we have maximal transfinite $\ices$s $\sigma_{\cutr(\mathfrak{G}'_i(j), \mathfrak{G}''_i(j))}$ and    $\sigma_{\mathfrak{G}'''_i(j)}$ satisfying the hypothesis. Since $\der$ is weakly regular  the sets of sequences  $X_i\dfn\{\sigma_{\cutr(\mathfrak{G}'_i(j), \mathfrak{G}''_i(j))}\ \vert \ j \geq 0 \}$ and    $Y_i \dfn \{\sigma_{\mathfrak{G}'''_i(j)}\ \vert \ j \geq 0\}$ can be assumed to be  finite. We set:
		$$
			\sigma_\der\dfn   (\sigma^\infty_\der(i))_{0 \leq i\leq n} \cdot \prod_{i=1}^m  \prod_{j=1}^\infty \sigma_{\mathfrak{G}'''_i(j)} \cdot \prod_{i=1}^{l} (\sigma_i \cdot \prod_{j=1}^{\infty} \sigma_{\cutr(\mathfrak{G}'_i(j), \mathfrak{G}''_i(j))})
		$$
		where $\sigma' \cdot \sigma''$ denotes the concatenation of two sequences $\sigma' $ and $ \sigma''$. Let us now show that $\sigma_\der$ satisfies~\Cref{enum:conditions-trasfinite1}-\ref{enum:conditions-trasfinite2}. This follows from the induction hypothesis and the construction of $\sigma_i$ ($1 \leq i \leq l$). Notice, indeed, that the $i$-th element of  $\sigma_i$ is the application of  a cut-elimination step to a  $\cutr$ with shape $\cutstep{\cprule}{\cprule}$ and with height $ i$. Clearly,~\Cref{enum:conditions-transfinite-cut-free} is satisfied.  Concerning~\Cref{enum:conditions-transfinite3},  since the sets of sequences $X_i$ and $Y_i$ are finite, using the induction hypothesis we have that if the sequences $\sigma_{X_i}\dfn  \prod_{j=1}^\infty \sigma_{\mathfrak{G}'''_i(j)}$ and $\sigma_{Y_i}\dfn \sigma_i \cdot \prod_{j=1}^{\infty} \sigma_{\cutr(\mathfrak{G}'_i(j), \mathfrak{G}''_i(j))}$ are applied to a weakly regular coderivation, their limit  is a weakly regular coderivation. From this fact and~\Cref{prop:cut-elim-preserves-finexp-reg-weakreg} we can conclude that the limit of $\sigma_\der$ is weakly  regular.
	\end{itemize}
	Now, let $\lim{} {\sigma_\der}$ be the limit of $\sigma _\der$. We want to show by induction $d$ that $\sigma$ can be rewritten to a $\mcices$ $\sigma^*$  such that  $\lim{}{\sigma_\der}= f_{\sigma^*}(\der)$.
	\begin{itemize}
		\item 	 The case $d=0$ follows by construction of $\sigma$. 
	
		\item Let us suppose $d >0$. By induction hypothesis we have  $\sigma^*_{\cutr(\mathfrak{G}'_i(j), \mathfrak{G}''_i(j))}$ and $\sigma^*_{\mathfrak{G}'''_i(j)}$ such that, 	for any $j \geq 0$:
		\begin{itemize}
			\item  $\lim{}{\sigma_{\cutr(\mathfrak{G}'_i(j), \mathfrak{G}''_i(j))}}= f_{\sigma^*}(\cutr(\mathfrak{G}'_i(j), \mathfrak{G}''_i(j)))$
			\item $\lim{}{\sigma_{\mathfrak{G}'''_i(j)}}= f_{\sigma^*}({\mathfrak{G}'''_i(j)})$.
		\end{itemize}
		Let us now  show that the sequences $\sigma_i \cdot \prod_{j=1}^{\infty} \sigma^*_{\cutr(\mathfrak{G}'_i(j), \mathfrak{G}''_i(j))}$ can be rewritten to a sequence with length $\omega$ with the same  limit and preserving conditions~\Cref{enum:conditions-transfiniteA}-\ref{enum:conditions-transfinite3}.  We notice that:
		\begin{itemize}
			\item for any $j \neq j'$,  cut-elimination steps in $\sigma^*_{\cutr(\mathfrak{G}'_i(j), \mathfrak{G}''_i(j))}$ commute with  cut-elimination steps  in $\sigma^*_{\cutr(\mathfrak{G}'_i(j'), \mathfrak{G}''_i(j'))}$ 
			\item   the $j+1$-th cut-elimination step of $\sigma_i$ commutes with all cut-elimination steps   in $\sigma(\cutr(\mathfrak{G}'_i(j'), \mathfrak{G}''_i(j')))$ with $j'<j$.
		\end{itemize}
		Since $\sigma_i$ is has length $\omega$, by the above observations, we define a  sequence $\sigma^*_{i}$ of length $ \omega$ divided into stages, where each stage consists of a finite subsequence of reduction steps. At the $n$-th stage:
		\begin{itemize}
			\item we apply the $n$-th cut-elimination step of $\sigma_i$ 
			\item for any $1 \leq j \leq n$ we apply (the next available) $n+1-j$  steps of  $\sigma^*_{\cutr(\mathfrak{G}'_i(j), \mathfrak{G}''_i(j))}$. 
		\end{itemize} 
		In a similar way, for any $1 \leq i \leq m$  the reduction sequence $\sigma^*_{\mathfrak{G}'''_i(j)} $ can be rewritten to a sequence $\sigma^{**}_{i}(\der)$ of length $\leq \omega$ (preserving the limit and conditions~\Cref{enum:conditions-transfiniteA}-\ref{enum:conditions-transfinite3}).  We obtain a sequence of the following form:
		$$
			(\sigma^\infty_\der(i))_{0 \leq i\leq n} \cdot  \prod_{i=1}^m   \sigma^{**}_i \cdot \prod_{i=1}^{l} \sigma^*_i
		$$
		Since any cut-elimination step in $\sigma^{**}_i$ commutes with any cut-elimination step in $\sigma^*_i$, we can rewrite the above  sequence to a  sequence $\sigma^*_\der=(\der_i)_i$ of length $\leq \omega$ with the same  limit and preserving conditions~\Cref{enum:conditions-transfiniteA}-\ref{enum:conditions-transfinite3}. By definition, to prove that $\lim{}{\sigma_\der}=f_{\sigma^*}(\der)$ it suffices to show that $\lim{}{\sigma_\der}=\bigsqcup_i \cf{\der_i}$:
		\begin{itemize}
			\item By~\Cref{enum:conditions-trasfinite1} we have $\lim{}{\sigma_\der}=\bigsqcup_i {\der'_i}$ for some $\der'_i$ approximations of $\der_i$ so that, by~\Cref{enum:conditions-transfinite-cut-free}, we have  $\lim{}{\sigma_\der}=\bigsqcup_i {\der'_i}\preceq \bigsqcup_i \cf{\der_i}$. 
			
			\item 	 By~\Cref{enum:conditions-trasfinite2} we  have  $\bigsqcup_i \cf{\der_i}\preceq \lim{}{\sigma_\der}$.
		\end{itemize}
	\end{itemize}

	This shows that $f_{\sigma^\infty}(\der)$ is weakly regular if $\der$ is. Therefore, if $\der \in \nupll$  then $f_{\sigma^\infty}(\der) \in \nupll$  by~\Cref{thm:preservation-weakly-progressive} and~\Cref{thm:preservation-progressiveness}.

	Concerning preservation of regularity,  we apply the same reasoning, checking that the $\ices$  preserves periodicity of $\nwbox$s.
\end{proof}

\section{Proofs of \Cref{sec:semantics}}

\semanticapproximation*
\begin{proof}
	By \Cref{prop:scott-domain}, $\der = \bigsqcup_{\der' \in \fapx[ \der]}\der'$.
		
	For the left-to-right inclusion, observe  that for every $n \in \Nset$ there is $\der'_n \in \fapx[ \der]$ such that
	$\sem{\bigsqcup_{\der'\in \fapx[ \der]} \der'}_n = \sem{\der'_n} \subseteq \bigcup_{\der' \in \fapx[ \der]}\sem{\der'}$.
	Therefore, by minimality of the union,
	$$
	\sem{\der}
	= 
	\bigcup_{n \in \Nset} \sem{\der}_n
	= 
	\bigcup_{n \in \Nset} \sem{\bigsqcup_{\der' \in \fapx[ \der]}\der'}_n 
	\subseteq 
	\bigcup_{\der' \in \fapx[ \der]}\sem{\der'}.
	$$
	
	As for the converse inclusion, we have that  $\der' \preceq \der''$ implies $\sem{\der'}\subseteq \sem{\der''}$.	
	Hence, for all $\der'\in \fapx[ \der]$, since $\der' \preceq \bigsqcup_{\der'\in \fapx[ \der]} \der' = \der$, we have $\sem{\der'} \subseteq \sem{\der}$.
	By minimality of the union, $\bigcup_{\der'\in \fapx[ \der]}\sem{\der'} \subseteq \sem{\der}$.
\end{proof}

\soundness*
\begin{proof}
\begin{enumerate}
	\item By straightforward inspection of the cut-elimination steps for $\opll$. 
	
	\item By definition of $\mcices$,  for any $\der' \in \fapx[\der]$ we have $\der' \cutelims f_\sigma(\der')$, so $\sem{\der'}= \sem{f_\sigma(\der')}$ by \Cref{thm:soundness-semantics}.\ref{p:soundness-semantics-single}.
	By \Cref{prop:scott-domain}, $\der = \bigsqcup_{\der' \in \fapx[ \der]}\der'$. 
	By continuity of $f_\sigma$, we have $f_\sigma(\der) = \bigsqcup_{\der' \in \fapx[ \der]}f_\sigma(\der')$.
	Therefore, by~\Cref{lem:semantic-approximation} we have:
	\begin{equation*}
		\adjustbox{max width=\textwidth}{$
			\sem{\der}
			=		
			\bigcup_{\der' \in \fapx[ \der]}\sem{\der'}
			=
			\bigcup_{\der' \in \fapx[ \der]}\sem{f_\sigma(\der')}
			=	
			\sem{\bigsqcup_{\der' \in \fapx[ \der]}f_\sigma(\der')}
			=
			\sem{f_\sigma(\der)}
			$.}
		\qedhere
	\end{equation*}
\end{enumerate}
\end{proof}

\nonempty*
\begin{proof}
	If $\der\in \wppll$ is a cut-free coderivation, then \wprog ensures the existence of a bar $\Nodes$ containing  conclusions of rules in $\set{\axr,\lone,\cprule}$. By weak K\"{o}nig's lemma,  $\prun{\der}{\Nodes}$ is finite. Then, we prove by induction on  $\prun{\der}{\Nodes}$ that there is $n \geq 0$ such that $\sem{\prun{\der}{\Nodes}}_n \neq \emptyset$, so that we conclude  $\emptyset\neq \sem{\prun{\der}{\Nodes}}_n \subseteq \sem{\der}_n\subseteq \sem{\der}$. 
	As for the base case, notice that the interpretation of any coderivation ending with the $\cprule$ contains the element $(\vec{\emptymset} , \emptymset)$, so it is never empty. The inductive steps are straightforward.
	
	If $\dD$ contains $\cutr$-rules, then
	$\sem{\der}=\sem{f_\sigma(\der)}$ 
	by \Cref{thm:soundness-semantics}. Since 	$f_\sigma(\der)$ is $\cutr$-free, we  conclude $\sem{\der}\neq\emptyset$ by the above reasoning.
\end{proof}

\end{document}